\documentclass[journal]{IEEEtran}

\ifCLASSINFOpdf
   \usepackage[pdftex]{graphicx}
\else
\fi

\usepackage{amsmath}
\usepackage{mathtools}

\usepackage{amsfonts}
\usepackage{amsthm}
\usepackage{thmtools}
\declaretheorem[name=Theorem]{theorem}

\declaretheorem[style=remark, numbered=no]{remark}

\usepackage{siunitx}

\ifCLASSOPTIONcompsoc
  \usepackage[caption=false,font=normalsize,labelfont=sf,textfont=sf]{subfig}
\else
  \usepackage[caption=false,font=footnotesize]{subfig}
\fi

\newcommand{\textsub}[1]{\textnormal{#1}}

\usepackage{bm}
\renewcommand{\vec}[1]{\bm{#1}}

\newcommand{\CW}{\mathcal{C}}

\newcommand{\dmin}{d_{\textsub{min}}}
\newcommand{\ddesign}{d_{\textsub{des}}}

\newcommand{\que}{\mathord{?}}

\newcommand{\ZO}{\{0, 1\}}
\newcommand{\ZQO}{\{0, \que, 1\}}

\DeclareMathOperator{\DF}{\mathsf{D}}
\newcommand{\DFC}{\DF_{\textsub{C}}}

\DeclareMathOperator{\dH}{d}

\DeclareMathOperator{\we}{w}

\newcommand{\Sp}{\mathcal{S}}

\newcommand{\EsNO}{\frac{E_{\textsub{s}}}{N_0}}
\newcommand{\lEsNO}{E_{\textsub{s}} / N_0}

\DeclareMathOperator{\QFunc}{Q}

\newcommand{\Topt}{T_{\textsub{opt}}}
\newcommand{\gainTheo}{\Delta(\lEsNO)^{\ast}}

\newcommand{\xor}[1]{\overline{#1}}

\newcommand{\compl}[1]{\mathord{\sim} #1}
\newcommand{\nE}{{\compl{E}}}
\newcommand{\E}{E}

\newcommand{\dE}[1]{\operatorname{d}_{E(#1)}}
\newcommand{\dnE}[1]{\operatorname{d}_{\compl{E(#1)}}}
\newcommand{\numE}{E}

\DeclareMathOperator{\Prob}{\mathbb{P}}

\newcommand{\Trans}[2]{(#1 \rightarrow #2)}
\newcommand{\T}[3]{T_{#1 \rightarrow #2}\left( #3 \right)}
\newcommand{\Ts}[2]{T_{#1 \rightarrow #2}}

\DeclareMathOperator{\Error}{Error}
\newcommand{\PError}{\Prob\left(\Error(\Dp, \Ep)\right)}

\newcommand{\Dp}{D^\prime}
\newcommand{\Ep}{E^\prime}

\newcommand{\Delp}{\Delta^\prime}

\usepackage{bbm}
\newcommand{\ind}[1]{\mathbbm{1}_{\{#1\}}}

\newcommand{\config}{\mathcal{K}}

\newcommand{\BER}{\mathrm{BER}}
\newcommand{\BEP}{\rho}

\newcommand{\OmegaP}{\Omega_{\textsub{p}}}

\newcommand{\deltaM}{\delta_{\textsub{m}}}
\newcommand{\epsilonM}{\epsilon_{\textsub{m}}}
\newcommand{\deltaC}{\delta_{\textsub{c}}}
\newcommand{\epsilonC}{\epsilon_{\textsub{c}}}
\newcommand{\chiM}{\vec{\chi}_{\textsub{m}}}
\newcommand{\chiC}{\vec{\chi}_{\textsub{c}}}

\newcommand{\Mset}{\mathcal{M}}

\DeclareMathOperator*{\argmin}{arg\,min}

\begin{document}
\title{Error-and-Erasure Decoding of Product and Staircase Codes}

\author{Lukas Rapp and Laurent Schmalen,~\IEEEmembership{Senior~Member,~IEEE}%
\thanks{This work has received funding from the European Research Council (ERC) under the European Union’s Horizon 2020 research and innovation programme (grant agreement No. 101001899)}
\thanks{Both authors are with Karlsruhe Institute of Technology (KIT), Communications Engineering Lab (CEL), Hertzstr. 16, 76187 Karlsruhe, Germany. E-mail: \texttt{lukas.rapp3@student.kit.edu}, \texttt{schmalen@kit.edu}}%
\thanks{\copyright 2021 IEEE. Personal use of this material is permitted. Permission from IEEE must be obtained for all other uses, in any current or future media, including reprinting/republishing this material for advertising or promotional purposes, creating new collective works, for resale or redistribution to servers or lists, or reuse of any copyrighted component of this work in other works. Article DOI: 10.1109/TCOMM.2021.3118789}}

\markboth{Accepted version, IEEE Transactions on Communications}%
{Rapp \MakeLowercase{\textit{et al.}}: Error-and-Erasure Decoding of Product and Staircase Codes}

\maketitle

\begin{abstract}
High-rate product codes (PCs) and staircase codes (SCs) are ubiquitous codes in high-speed optical communication achieving near-capacity performance on the binary symmetric channel. Their success is mostly due to very efficient iterative decoding algorithms that require very little complexity.
In this paper, we extend the density evolution (DE) analysis for PCs and SCs to a channel with ternary output and ternary message passing, where the third symbol marks an erasure. We investigate the performance of a standard error-and-erasure decoder and of its simplification using DE. The proposed analysis can be used to find component code configurations and quantizer levels for the channel output. We also show how the use of even-weight BCH subcodes as component codes can improve the decoding performance at high rates. The DE results are verified by Monte-Carlo simulations, which show that additional coding gains of up to 0.6\,dB are possible by ternary decoding, at only a small additional increase in complexity compared to traditional binary message passing.
\end{abstract}


\IEEEpeerreviewmaketitle

\section{Introduction}
The implementation of high-speed communications is a challenging task. Commercially available transceivers for optical communications operate at throughputs of $800$\,Gbit/s and beyond~\cite{sun2020800g}. In order to maximize throughput and transmission reach, powerful forward error correction (FEC) is necessary. Modern FEC schemes require net coding gains of 11\,dB and more at residual bit error rates (BERs) of $10^{-15}$, for code rates larger than $0.8$~\cite{graell20forward}. For high-performance applications, soft-decision decoding (SDD) of low-density parity-check (LDPC) codes is now state-of-the-art in fiber-optic communication (see, e.g.,~\cite{graell20forward} for further references and~\cite{sun2020800g} for a recent commercial example). The adoption of SDD in fiber-optic communications represented a breakthrough with respect to the classical schemes based on algebraic codes (BCH and Reed-Solomon codes) and hard-decision decoding. However, the implementation of SDD schemes for popular codes still presents several challenges at very high data rates, in particular due to large internal decoder data flows~\cite{Smith2012}. Recently,  optimized codes for SDD with reduced decoder dataflows were proposed~\cite{barakatain2018low}, but these schemes require an additional low-complexity outer code (the latter being subject of the investigations in this paper).  

Some ubiquitous applications like data-center inter- and intraconnects require an extremely low transceiver complexity, which leads to heavy power consumption constraints on the transceiver circuits that often prohibit the use of SDD. The lower complexity of typical hard-decision decoding (HDD) circuits  motivates their use for applications where complexity and throughput is a concern~\cite{Smith2012}.  Powerful code constructions for HDD date back to the 1950s, when Elias introduced product codes~\cite{Elias54}. In the recent years, the introduction of new code constructions, such as staircase codes~\cite{Smith2012} and related schemes~\cite{Jian2013}, \cite{sukmadji2019zipper}, and the link between these constructions and codes-on-graphs, has led to a renewed interest in HDD for high-speed communications. 

HDD unfortunately entails an unavoidable capacity loss stemming from the hard decision at the channel output, reducing the achievable coding gains by 1-2\,dB compared to SDD. Recent work has focused on improving the performance of modern codes for HDD by employing soft information from the channel, see, e.g.,~\cite{graell20forward,hager2018approaching,sheikh2019binary,sheikh2020novel} and references therein. Most of these schemes assume that the decoder has access to the full soft information (e.g., the channel output after transmission over a binary-input additive white Gaussian noise (AWGN) channel model) and internally use binary or ternary message passing~\cite{lechner2011analysis,yacoub2019protograph} and possibly error-and-erasure decoding~\cite{forney1966generalized,Wicker95,Blahut} of the component codes. However, in many high-speed optical communication systems, in particular those optimized for low cost and short reach, the use of a high-precision analog-to-digital converter (ADC) is prohibitive as the power consumption of an ADC scales approximately in proportion to its bit resolution~\cite{pillai2014end} and often simple 1-bit ADCs are used~\cite{ossieur2018asic}. 

A promising approach for reducing the capacity loss while still keeping both the receiver and decoding complexity low is error-and-erasure decoding of linear codes using a 3-level (ternary) ADC at the channel output.
For instance, error-and-erasure decoding can be implemented by just two usual binary decodings and a little decision logic \cite{MoonBook}.
While error-and-erasure decoding for algebraic and product codes~\cite{wainberg1972error} is well understood, its application to modern codes for high-speed communications is largely unexplored. 
The ternary output increases the capacity of the binary-output channel and can be used to improve decoding of, e.g., LDPC codes~\cite{RichardsonCapa,yacoub2019protograph}. 
Recently, it was shown using both simulations and a stall pattern analysis that error-and-erasure decoding for product or staircase codes can improve their decoding performance~\cite{sumkmadji2020zipper,soma2021errors}. A rigorous analysis including miscorrections and allowing easy parameter optimization was however lacking.

In this paper, we investigate the potential of ternary message passing with ternary channel outputs for high-rate product and staircase codes with BCH component codes. Our investigation extends the density evolution analysis of~\cite{jianApproachingCapacity2017} to ternary channel outputs and ternary message passing for various decoding algorithms. This analysis fully takes into account possible miscorrections. One goal of the analysis is to find the quantizer levels that maximize the decoding performance. Interestingly, we find that the optimal quantizer, for which the noise threshold gets minimal, is significantly different from the one maximizing the capacity and that the gains of the noise threshold that can be obtained are less than the maximum achievable capacity gain.

\section{Background}\label{sec:background}
\subsection{Product \& Staircase Codes}
\subsubsection{Component Codes}
In this paper, \(\CW\) denotes a linear \((n, k, t)\) component code of a product or staircase code that is decoded by a \(t\) error-correcting bounded distance decoder (BDD), as described in Sec. \ref{sec:Decoder}.

For product codes,  \(\CW\) is either a \((2^\nu - 1, k, t)\) binary cyclic BCH code, \(\CW_{\textsub{BCH}}\), or its \((2^\nu - 1, k - 1, t)\) cyclic even-weight subcode,
\(
    \CW_{\textsub{BCH-Ev}} \coloneqq \{\vec{c} \in \CW_{\textsub{BCH}} : \we(\vec{c}) = 2 j \}
\).
Although the minimum distance of BCH codes is in general not known, a lower bound is given by the design distance, \(\ddesign(t)\), which is \(2 t + 1\) for a BCH code and \(2 t + 2\) for its even-weight subcode.

For staircase codes, we use shortened BCH codes or shortened even-weight subcodes, i.e. we take from an \((n, k, t)\) code only the codewords \(\vec{c}\) that begin with \(c_1 = 0\) and delete the first coordinate
\cite[Ch.~1. §9]{MacWilliamsSloane}. By doing so, we obtain an \((n - 1, k - 1, t)\) linear code.
The \(\ddesign(t)\) is \(2 t + 1\) for a shortened BCH code and \(2 t + 2\) for a shortened even-weight subcode.

\subsubsection{Product Code}
A product code of an \((n, k, t)\) component code \(\CW\) is a set of binary \(n {\times} n\) matrices whose rows and columns are codewords of \(\CW\), resulting in a code of rate \(r = \left(\frac{k}{n}\right)^2\). To decode a product code, the rows and columns are alternately decoded by the component decoder \(\DFC\). %
A product code can be interpreted as a generalized LDPC (GLDPC) code, hence, its performance under iterative decoding can be estimated through the average performance of a proper GLDPC ensemble.
This makes an analysis via density evolution (DE) possible as described in \cite{jianApproachingCapacity2017}.
The adequate GLDPC ensemble consists of the Tanner graphs with \(m\) constraint nodes (CNs) of degree \(n\) and \(N = \frac{n m}{2}\) variable nodes (VNs) of degree \(2\). In the following, the ensemble is denoted as \((\CW, m)\) GLDPC ensemble.

The CNs of the Tanner graphs are defined by \(\CW\), i.e., the binary values of the VNs connected to a CN must form a valid codeword of \(\CW\).
To construct a random graph of this ensemble, the outgoing edges of the VNs are connected to the sockets of the CNs via a random permutation~\cite{jianApproachingCapacity2017}.%

\subsubsection{Staircase Code}
A staircase code of an \((n, k, t)\) component code \(\CW\) is a chain of \(L\) binary matrices of size \(\frac{n}{2} {\times} \frac{n}{2}\). Its rate is \(r = 2\frac{k}{n} - 1\) \cite{Smith2012}.
Similar to the product code, we consider the \((\CW, m, L)\) spatially-coupled GLDPC (SC-GLDPC) ensemble for the analysis.
\begin{figure}
    \centering
    \includegraphics{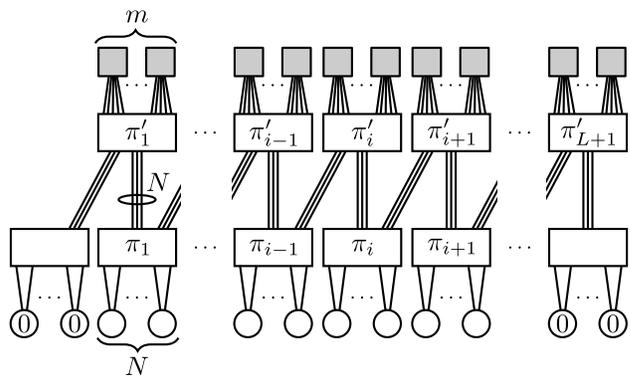}
    \caption{Random element of the \((\CW, m, L)\) SC-GLDPC ensemble. \(\pi_i\) and \(\pi^\prime_i\) are random permutations of the edges. Image based on~\cite{jianApproachingCapacity2017}.}
    \label{fig:SCGLDPCCode}
\end{figure}
Figure~\ref{fig:SCGLDPCCode} shows the construction of a random Tanner graph of this ensemble.
In the ensemble, the VNs are divided into \(L\) groups and the CNs into \(L + 1\) groups. Each group of VNs contains \(N = \frac{n m}{2}\) nodes of degree \(2\) and each group of CNs contains \(m\) nodes of degree \(n\) so that each group of VNs or CNs has \(2 N\) edges.
To construct a random Tanner graph, the \(2 N\) edges of each group are divided via a uniform random permutation \(\pi_i\) and \(\pi_i^\prime\), respectively, into two sets of \(N\) edges. The first set of edges of VN group \(i \in \{1, \dotsc, L\}\) is connected to a set of edges of CN group \(i\) and the second set is connected to a set of edges of CN group \(i + 1\). 
The remaining edges of CN group \(1\) and \(L + 1\) are connected to VNs with the fixed value \(0\), which can be shortened.

\subsubsection{GLDPC Decoding}
The GLDPC codes of both ensembles are decoded via the same message passing algorithm, which we briefly explain here. The CNs and VNs of the Tanner graph are indexed. Let \(\sigma_j(k)\) be the index of the VN that is connected to socket \(k \in \{1, \dotsc, n\}\) of the \(j\)-th CN.
During message passing, the messages belonging to a set \(S\) are passed along the edges between VNs and CNs. For HDD, the messages are from \(S = \ZO\) and for the error-and-erasure decoding introduced below, \(S = \ZQO\). 
Let \(\nu_{i,j}^{(\ell)} \in S\) be the message that is passed from the \(i\)-th VN to the \(j\)-th CN in the \(\ell\)-th iteration and let \(\tilde{\nu}_{i, j}^{(\ell)}\) be the message that is passed back from CN \(j\) to VN \(i\) in the \(\ell\)-th iteration.
To decode a received word \(\vec{r} =(r_1, r_2, \ldots)\), where \(r_i \in S\) is the received channel value of the \(i\)-th VN, the following steps are performed:
During initialization, the received channel value, \(r_i\), of each VN \(i\) is sent to its two connected CNs, \(j, j^\prime\), where we set \(\nu_{i, j}^{(1)} = \nu_{i, j^\prime}^{(1)} = r_i\).
Then, several decoding iterations are performed consisting of a CN update followed by a VN update.

In the \(\ell\)-th CN update, each CN \(j\) receives the incoming messages
\((\nu_{\sigma_j(1), j}^{(\ell)}, \dotsc, \nu_{\sigma_j(n), j}^{(\ell)})\). To calculate the message that is sent back to the VN \(i, \sigma_j(k)\) connected at the \(k\)-th position of CN \(j\), two different approaches are considered: Intrinsic message passing (IMP)~\cite{Smith2012} and extrinsic message passing (EMP)~\cite{jianApproachingCapacity2017}.
For IMP, the incoming messages are combined to the word
\[
    \vec{y}_{j, \textsub{IMP}}^{(\ell)} 
    \coloneqq
    (\nu_{\sigma_j(1), j}^{(\ell)}, \dotsc, \nu_{\sigma_j(n), j}^{(\ell)})
\]
and are decoded by the component decoder \(\DFC\). Then, the \(k\)-th symbol of the result is sent back to VN \(i\): 
\(
    \tilde{\nu}_{i, j}^{(\ell)} 
    = 
    \big[\DFC(\vec{y}_{j, \textsub{IMP}}^{(\ell)})\big]_k
\).
For EMP, the \(k\)-th incoming message is replaced by the channel value \(r_i\) resulting in the extrinsic word
\[
    \phantom{.}
    \vec{y}_{j, \textsub{EMP}, k}^{(\ell)} 
    \coloneqq
    (
        \nu_{\sigma_j(1), j}^{(\ell)}, \dotsc ,
        \nu_{\sigma_j(k-1), j}^{(\ell)}, r_{i}, 
        \nu_{\sigma_j(k+1), j}^{(\ell)}, \dotsc
    )
    .
\]
Then, the word is decoded and the \(k\)-th symbol of the result is sent back to VN \(i\):
\(
    \tilde{\nu}_{i, j}^{(\ell)} 
    = 
    \big[\DFC(\vec{y}_{j, \textsub{EMP}, k}^{(\ell)})\big]_k
\).

In the VN update, each VN \(i\) receives two messages from its connected CNs \(j\), \(j^\prime\)
and forwards to each CN the message that it has received from the respective other CN:
\(\nu_{i, j^\prime}^{(\ell+1)} = \tilde{\nu}_{i, j}^{(\ell)}\),
\(\nu_{i, j}^{(\ell+1)} = \tilde{\nu}_{i, j^\prime}^{(\ell)}\).

At the end of the message passing, each VN has two incoming messages to determine the decoding result. To make a decision, one of the incoming messages is chosen randomly. If the message is erased, it is replaced by a random binary value.%
\footnote{Note that this decision rule is not optimal. In practical decoders, one would only choose randomly if both messages are erased. We use the proposed rule because it allows an easy calculation of the final bit error probability in the DE (see \eqref{eqn:Def_BER}).}

\subsubsection{Remarks}
In practice, a sliding window is used to decode a staircase code. In most cases, to the best of our knowledge, IMP is used due to the lower memory requirements. This window slides over the binary matrices and decoding is only performed for matrices in the window \cite{Smith2012}. We neglect windowed decoding in our analysis and our results can be seen as an upper bound on the performance under windowed decoding.

Note that EMP requires \(n\) component decodings per CN update whereas IMP requires only one. 
However, for EMP decoding without erasures, there exist an algorithm that requires only one decoding \cite{jianApproachingCapacity2017}. Hence, the complexity does not increase by the factor \(n\) because CN updates without erasures can be carried out with this algorithm and the number of erasures is normally very low after only a few iterations.

Further note that we restrict our analysis to the GLDPC and SC-GLDPC code ensembles. Product and staircase codes are not necessarily typical code realizations of these ensembles, hence the analysis may not directly apply. Numerical investigations show however good agreements between the ensemble analysis and the decoding performance of product and staircase codes \cite[Sec. 7.5.9]{graell20forward}, \cite{hager2018approaching}. The behavior of more deterministic code constructions has been analyzed in~\cite{hager2017density} and~\cite{zhang2017spatially} for the binary erasure channel (using two different approaches), but the authors acknowledge that their approach cannot be easily extended towards more general channels without ignoring miscorrections.

\subsection{Error-and-Erasure Decoding}
\subsubsection{Channel} \label{sec:channel}
For the following analysis, we assume that the GLDPC codewords \(\vec{x}\) are transmitted over a binary-input additive white Gaussian noise (BI-AWGN) channel which generates 
\(
    \tilde{r}_i \coloneqq (-1)^{x_i} + n_i
\), 
where \(n_i\) is an AWGN sample with noise variance
\(
    \sigma^2 = (2 \lEsNO)^{-1}
\).
To reduce the capacity loss due to HDD, error-and-erasure decoding uses a ternary channel output and message alphabet \(\ZQO\). To determine the discrete channel outputs \(r_i\), the values \(\tilde{r}_i \in [-T, +T]\) are declared as erasure ``\(\que\)''. Values outside this interval are mapped to \(0\) and \(1\) by the usual HDD rule, i.e. \(r_i = 1\) for \(\tilde{r}_i < -T\) and \(r_i = 0\) for \(\tilde{r}_i > +T\).

This channel is abstracted through the discrete, memoryless channel model shown in Fig.~\ref{fig:channel_model}.
\begin{figure}
    \centering
    \includegraphics{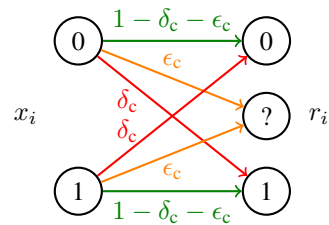}
    \caption{Discrete channel model}
    \label{fig:channel_model}
\end{figure}
The channel transition probabilities are given by
\begin{equation}\label{eqn:Channel_Trans_Probs}
\begin{aligned}
    \operatorname{\deltaC}
    &= 
    \QFunc\left(\sqrt{2\EsNO} (T + 1)\right),\\
    \operatorname{\epsilonC}
    &= 
    1 -  
    \QFunc\left(\sqrt{2\EsNO} (T - 1)\right) - 
    \QFunc\left(\sqrt{2\EsNO} (T + 1)\right),
\end{aligned}
\end{equation}
where \(\deltaC\) is the probability for an error and \(\epsilonC\) for an erasure.
Since the channel is completely described through \(\lEsNO\) and \(T\), it is denoted by \((\lEsNO, T)\).

It is easy to see that for a fixed $T$, the capacity of this channel is
\begin{equation*}
    \phantom{,}
    C\left(\EsNO, T\right) 
    = 
    c_{\textsub{c}} \log_2\left( \frac{2 c_{\textsub{c}}}{1 - \epsilonC} \right) 
    + \deltaC \log_2 \left( \frac{2 \deltaC}{1 - \epsilonC} \right)
    ,
\end{equation*}
where \(c_{\textsub{c}} \coloneqq 1 - \deltaC - \epsilonC\) is the probability of correctly receiving a symbol.

Optimization of \(C(\lEsNO, T)\) with respect to \(T\) results in a capacity gain for this channel compared to HDD (\(T=0\)). 

\subsubsection{Decoder} \label{sec:Decoder}
The decoder of the introduced component codes \(\CW\) is a bounded distance decoder (BDD).
Let
\[
    \Sp_t(\vec{c}) \coloneqq \{\vec{y} \in \ZO^n : \dH(\vec{y}, \vec{c}) \leq t\}
\]
be the Hamming sphere of radius \(t\) around a codeword \(\vec{c} \in \CW\) that consists of all words \(\vec{y} \in \ZO^n\) whose Hamming distance from \(\vec{c}\) is less than or equal to \(t\). 
For a given word \(\vec{y} \in \ZO^n\), a \(t\) error-correcting BDD selects the codeword \(\vec{c} \in \CW\) for which \(\vec{y} \in \Sp_t(\vec{c})\) holds. Otherwise, a decoding failure is declared:
\[
    \DF_{\textsub{BDD}}(\vec{y}) \coloneqq
    \begin{cases}
        \vec{c} & \text{if \(\exists \vec{c}\in\CW\) such that \(\vec{y} \in \Sp_t(\vec{c})\)} \\
        \text{fail} & \text{otherwise}.%
    \end{cases}
\]

Since the channel output alphabet is \(\ZQO\), a BDD cannot be used. Hence, we use the following error-and-erasure decoder (EaED), which is a modification of \cite[Sec.~3.8.1]{MoonBook}. %
Let \(\numE(\vec{y}) = |\{i \in \{1, \dotsc, n\} : y_i = \que\}|\) be the number of erasures of the word \(\vec{y}\) and let \(\dnE{\vec{y}}(\vec{a}, \vec{b})\) be the Hamming distance between the words \(\vec{a}\) and \(\vec{b}\) at the unerased coordinates of \(\vec{y}\).
The EaED performs the following steps to decode a word \(\vec{y} \in \ZQO^n\) to the result \(\vec{w}\):
\begin{enumerate}
    \item If \(\numE(\vec{y}) \geq \ddesign(t)\), \(\vec{w} = \vec{y}\). Otherwise, continue with \ref{item:GenerateY1}).
    \item \label{item:GenerateY1} 
    Generate a random vector \(\vec{p} \in \ZO^{\numE(\vec{y})}\) and place the values of \(\vec{p}\) at the erased coordinates of \(\vec{y}\), yielding \(\vec{y}_1\).
    \item \label{item:GenerateY2} 
    Generate the inverted vector of \(\vec{p}\), denoted by \(\xor{\vec{p}}\), by inverting every bit of \(\vec{p}\) and placing the values of \(\xor{\vec{p}}\) at the erased coordinates of \(\vec{y}\), yielding \(\vec{y}_2\).
    \item Decode $\vec{y}_i$, $i\in\{1,2\}$, using the BDD: 
    \(\vec{w}_i = \DF_{\textsub{BDD}}(\vec{y}_i)\) %
    \item \label{item:Decision_EaED} 
    Obtain the decoding result, \(\vec{w}\), as
    \begin{LaTeXdescription}%
        \item[Case 1:] 
        \(\vec{w}_1 = \vec{w}_2 = \text{fail}\): \(\vec{w} = \vec{y}\)
        \item[Case 2:] 
        \(\vec{w}_i \in \CW \text{ for exactly one \(\vec{w}_i\)}\): \(\vec{w} = \vec{w}_i\)
        \item[Case 3:]
        \(\vec{w}_1, \vec{w}_2 \in \CW\):
        Output the codeword \(\vec{w}_i\) for which \(\dnE{\vec{y}}(\vec{y}, \vec{w}_i)\) is smallest.
        If both distances are equal, one codeword \(\vec{w}_i\) is chosen at random.
    \end{LaTeXdescription}
\end{enumerate}
\begin{remark}
In practical decoders, \(\vec{p}\) is usually the all-zero vector. However, this is not suitable for our analysis, based on the all-zero codeword, because the decoder preferably decodes to the all-zero codeword leading to a falsified too good analysis result.
The use of random vectors in step 2), akin to the channel adapters of~\cite{hou2003capacity}, solves this issue which we prove in Theorem~\ref{thm:Performance_CW_Independent}.
\end{remark}

The following theorem, based on \cite[Sec.~3.8.1]{MoonBook}, %
estimates the correction capability of the EaED:
\begin{theorem}
\label{thm:EaED_correcting_capability}
    For the defined component codes, the EaED will correct a word with \(D\) errors and \(E\) erasures for certain if
    \begin{equation}\label{eqn:Condition_EaED_correcting_capability}
        2 D + E < \ddesign(t).
    \end{equation}
\end{theorem}
\begin{proof} 
    See Appendix~\ref{proof:EaED_correcting_capability}.
\end{proof}

In addition, we consider a simplification of the EaED. For this, %
we define the Hamming spheres in \(\ZQO^n\) as
\[
    \Sp^3_t(\vec{c}) = \{\vec{y} \in \ZQO^n : 2 \dnE{\vec{y}}(\vec{y}, \vec{c}) + \numE(\vec{y}) < \ddesign(t)\}.
\]
The extended EaED (EaED+) is then given by
\[
    \DF_{\textsub{EaED+}}(\vec{y}) \coloneqq
    \begin{cases}
        \vec{w} \coloneqq \DF_{\textsub{EaED}}(\vec{y}) &\text{if \(\vec{w} \in \CW\) and \(\vec{y} \in \Sp^3_t(\vec{w})\)}\\
        \vec{y} & \text{otherwise}.
    \end{cases}
\]
Because of Theorem \ref{thm:EaED_correcting_capability} and the linearity of \(\CW\), the EaED decodes deterministically all \(\vec{y} \in \Sp_t^3(\vec{c})\) to a codeword \(\vec{c}\). Hence, the EaED+ decodes a word \(\vec{y}\) to a codeword \(\vec{c}\) if and only if \(\vec{y} \in \Sp_t^3(\vec{c})\). This leads to an alternative definition of the EaED+, which is used in the following analysis:
\[
    \phantom{.}
    \DF_{\textsub{EaED+}}(\vec{y})
    =
    \begin{cases}
        \vec{c} & \text{if \(\exists\vec{c}\in\CW\) such that \(\vec{y} \in \Sp_t^3(\vec{c})\)} \\
        \vec{y} & \text{otherwise}.
    \end{cases}
\]

\begin{remark}
In contrast to the EaED+, the EaED will also decode error patterns outside the Hamming spheres with a certain probability. This allows the correction of more errors but there will be also more miscorrection for patterns with too many errors. We will see later decoding configuration in which each decoder outperforms the other one.
\end{remark}

\section{Density Evolution}\label{sec:density_evolution}
In the following, we assume that the all-zero codeword is transmitted, which is justified by the following theorem:
\begin{theorem}
\label{thm:Performance_CW_Independent}
    The performance of the GLDPC decoder is independent of the transmitted codeword for all introduced component decoders.
\end{theorem}
\begin{proof} 
    See Appendix~\ref{proof:Performance_CW_Independent}.
\end{proof}

To analyze the decoding performance of a product or staircase code, we
analyze the average performance of the corresponding GLDPC ensemble by DE.
For the analysis, we assume that the codewords are transmitted over a channel \((\lEsNO, T)\) and EMP is used. \(\chiC \coloneqq (\deltaC, \epsilonC)\) denotes the channel transition probabilities, which are calculated using \eqref{eqn:Channel_Trans_Probs}.

\subsection{GLDPC Ensemble}
As shown in \cite{jianApproachingCapacity2017}, the \((\CW, m)\) GLDPC ensemble can be analyzed by DE if the limit \(m \to \infty\) is considered.%
\footnote{%
It is not immediately obvious that the proposed EMP allows DE. The explanation for this is given in Appendix~\ref{sec:DE_EMP}.}
Let \(\chiM^{(\ell)} \coloneqq (\deltaM^{(\ell)}, \epsilonM^{(\ell)})\) be the error and erasure probability of the VN-to-CN messages \(\nu^{(\ell)}_{i, j}\) in the \(\ell\)-th iteration. 
In the first iteration, we have \(\chiM^{(1)} = \chiC\) because the VN-to-CN messages are initialized with the received channel values.

To derive the DE recursion, we randomly select a VN \(i\), which is connected to a CN \(j\) at position \(k = \sigma_j^{-1}(i)\) and to a second CN \(j^\prime\). Now, we consider the message \(\tilde{\nu}^{(\ell)}_{i, j}\) that is passed from CN \(j\) to VN \(i\) in the \(\ell\)-th iteration.
To compute this message, CN \(j\) constructs
\(
    \vec{e} 
    \coloneqq 
    \vec{y}_{j, \textsub{EMP}, k}^{(\ell)}
\).
By definition, \(e_k\) is replaced by \(r_i\), hence, the error and erasure probabilities of \(e_k\) are \(\chiC\). The other positions of \(\vec{e}\) are VN-to-CN messages, which are wrong or erased with the probabilities \(\chiM^{(\ell)}\).
We will call these positions ``\(\compl{k}\)'' with \(\compl{k} \subset \{1, \dotsc, n\}\) in the following.
After construction, \(\vec{e}\) is decoded to \(\vec{w} \coloneqq \DFC(\vec{e})\), and the \(k\)-th symbol \(w_k\) is sent to VN \(i\) and forwarded to CN \(j^\prime\):
\(\nu^{(\ell+1)}_{i, j^\prime} = \tilde{\nu}^{(\ell)}_{i, j} = w_k\).
This leads to the DE recursion
\begin{equation}
\label{eqn:GLPDC_Rec}
\begin{aligned}
    \phantom{.}
    \chiM^{(\ell + 1)}
    &=
    \chi_{\textsub{rec}}(\chiM^{(\ell)})
    \coloneqq
    (
        \operatorname{\delta}_{\textsub{rec}}(\chiM^{(\ell)}),
        \operatorname{\epsilon}_{\textsub{rec}}(\chiM^{(\ell)})
    )\\
    &=
    (\Prob(w_k = 1), \Prob(w_k = \que)),
\end{aligned}
\end{equation}
which is a system of two coupled recursive functions. %

Next, we decompose these probabilities. We define the event
\[
    \Error(\Dp, \Ep) 
    \coloneqq 
    \{\text{%
        \(\vec{e}\) has 
        \(\Dp\) \(1\)s and 
        \(\Ep\) \(\que\)s 
        in \(\compl{k}\)%
    }\}
\]
with the probability
\begin{multline*}
    f(\Dp, \Ep, \chiM^{(\ell)}) \coloneqq \PError =\\
    \binom{n - 1}{\Dp, \Ep}
    \left(\deltaM^{(\ell)}\right)^{\Dp}
    \left(\epsilonM^{(\ell)}\right)^{\Ep}
    \left(1 - \deltaM^{(\ell)} - \epsilonM^{(\ell)}\right)^{n - 1 - \Dp - \Ep}
    ,
\end{multline*}
where \(\binom{n - 1}{\Dp, \Ep} := \frac{(n-1)!}{\Dp!\Ep!(n-1-\Dp-\Ep)!}\) is the multinomial coefficient counting the ways of distributing \(\Dp\) \(1\)s and \(\Ep\) \(\que\)s in \(n-1\) positions.
Furthermore, we define the decoder transition probabilities
\begin{equation*}
    \T{\alpha}{\beta}{\Dp, \Ep} 
    \coloneqq
    \Prob\left(w_k = \beta \mid %
    \text{%
        \(e_k = \alpha\),
        \(\Error(\Dp, \Ep)\)%
    }\right)
\end{equation*}
which depend on the respective component decoder. We will determine these probabilities in Sec.~\ref{sec:calc_dec_trans_probs}.
Applying the law of total probability two times to \(\Prob(w_k = 1)\) results in
\begin{multline*}
    \operatorname{\delta}_{\textsub{rec}}(\chiM^{(\ell)})
    = 
    \Prob(w_k = 1) =
    \sum_{\Dp = 0}^{n-1}
    \smashoperator[r]{\sum_{\Ep = 0}^{n-1-\Dp}}
    f(\Dp, \Ep, \chiM^{(\ell)})\\ 
     \Big(\deltaC \T{1}{1}{\Dp, \Ep}
    + \epsilonC \T{\que}{1}{\Dp, \Ep} 
    + c_{\textsub{c}} \T{0}{1}{\Dp, \Ep} \Big)
    ,
\end{multline*}
where \(c_{\textsub{c}} \coloneqq 1 - \deltaC - \epsilonC\).
A similar decomposition is possible for \(\Prob(w_k = \que)\) leading to %
\[
    \operatorname{\epsilon}_{\textsub{rec}}(\chiM^{(\ell)})
    = 
    \sum_{\Dp = 0}^{n-1}
    \smashoperator[r]{\sum_{\Ep = 0}^{n-1-\Dp}}
    f(\Dp, \Ep, \chiM^{(\ell)})
    \epsilonC \T{\que}{\que}{\Dp, \Ep},
\]
where we used \(\T{0}{\que}{\Dp, \Ep} = \T{1}{\que}{\Dp, \Ep} = 0\) for all \(\Dp\) and \(\Ep\) because these transitions do not occur with the selected decoders. 

\begin{remark}
With some adjustments, it is possible to analyze codes with different component codes. For instance, for a product code with different code types for rows and columns, two different DE recursions could be applied one after the other. However, note that the degree of the VNs in the ensemble should still be \(2\) to enable the simplified VN update.
\end{remark}

\subsection{SC-GLDPC Ensemble}
To take the structure of the SC-GLDPC ensemble into account, error and erasure probabilities are defined for the messages of each VN or CN group corresponding to different edge types.
Let 
\(\vec{\chi}_{\textsub{m},i}^{(\ell)} = (\delta_{\textsub{m}, i}^{(\ell)}, \epsilon_{\textsub{m}, i}^{(\ell)})\)
be the average error and erasure probability of the messages that are sent in the \(\ell\)-th iteration from the VNs of group \(i\) to the CNs. The values of the VNs in group \(0\) and \(L + 1\) are fixed and therefore known at the decoder. Hence, their messages are always correct:
\(\vec{\chi}_{\textsub{m}, 0}^{(\ell)} = \vec{\chi}_{\textsub{m}, L + 1}^{(\ell)} = (0, 0)\).
The average error and erasure probability \(\vec{\hat{\chi}}_{\textsub{m}, i}^{(\ell)}\) of the messages sent to the CNs of group \(i\) is
\(
    \vec{\hat{\chi}}_{\textsub{m}, i}^{(\ell)}
    =
    \frac{1}{2}
    (
        \vec{\chi}_{\textsub{m}, i - 1}^{(\ell)} + 
        \vec{\chi}_{\textsub{m}, i}^{(\ell)}
    )
\)
because half of the messages are from VN group \(i - 1\) and the other half are from VN group \(i\) (see Fig.~\ref{fig:SCGLDPCCode}).
At the CNs, the CN update is performed and the CNs of group \(i\) return messages with the probabilities 
\(
    \operatorname{\chi}_{\textsub{rec}}\big(
        \vec{\hat{\chi}}_{\textsub{m}, i}^{(\ell)}
    \big)
\)
to the VNs. \(\operatorname{\chi}_{\textsub{rec}}\) denotes the DE recursion of the GLDPC ensemble as defined in \eqref{eqn:GLPDC_Rec}.

Then, at VN group \(i\), the probabilities \(\vec{\chi}_{\textsub{m}, i}^{(\ell+1)}\) are derived by averaging over the probabilities of the messages which are sent to this group in the last iteration.
This leads to the recursion
\begin{align} \label{eqn:SC_GLPDC_Rec}
    &\vec{\chi}_{\textsub{m}, i}^{(\ell+1)}
    = 
    \frac{1}{2}
    \big(
    \operatorname{\chi}_{\textsub{rec}}\big(
        \vec{\hat{\chi}}_{\textsub{m}, i}^{(\ell)}
    \big) +
    \operatorname{\chi}_{\textsub{rec}}\big(
        \vec{\hat{\chi}}_{\textsub{m}, i+1}^{(\ell)}
    \big)
    \big)\\
    &=
    \frac{1}{2}
    \left(
        \operatorname{\chi}_{\textsub{rec}}\bigg(
            \frac{\vec{\chi}_{\textsub{m}, i-1}^{(\ell)} + \vec{\chi}_{\textsub{m}, i}^{(\ell)}}{2}
        \bigg) + 
        \operatorname{\chi}_{\textsub{rec}}\bigg(
            \frac{\vec{\chi}_{\textsub{m}, i}^{(\ell)} + \vec{\chi}_{\textsub{m}, i+1}^{(\ell)}}{2}
        \bigg)
    \right)
    \nonumber
    ,
\end{align}
for \(i = 1, \dotsc, L\).

\subsection{Calculation of the Decoder Transition Probabilities}
\label{sec:calc_dec_trans_probs}
In this section, we calculate \(\T{\alpha}{\beta}{\Dp, \Ep}\) for both decoders.
For this, we only consider \(\Ts{\alpha}{\beta}\) with \(\alpha \neq \beta\).
The required transitions with \(\alpha = \beta\) are given by
\begin{align*}
    \T{1}{1}{\Dp, \Ep} &= 
    1 - \T{1}{0}{\Dp, \Ep} - \T{1}{\que}{\Dp, \Ep}, \\
    \T{\que}{\que}{\Dp, \Ep} &= 
    1 - \T{\que}{0}{\Dp, \Ep} - \T{\que}{1}{\Dp, \Ep}.
\end{align*}
Since the transition \(\Trans{1}{\que}\) does not happen, we only need to compute
\(\Ts{0}{1}\),\(\Ts{\que}{1}\),\(\Ts{1}{0}\) and \(\Ts{\que}{0}\).

For \(E \coloneqq \Ep + \ind{\alpha=\que} \geq \ddesign(t)\), with \(\ind{\alpha=\que}\) denoting the indicator function returning $1$ if the condition \(\{\alpha=\que\}\) is true and $0$ otherwise, both decoders return the input word unchanged. This results in \(\T{\alpha}{\beta}{\Dp,\Ep} = 0\) for \(\alpha \neq \beta\). Hence, in the following, only the cases with \(E < \ddesign(t)\) are considered.

\subsubsection{Weight Distributions}
For the following calculations, we require the weight distributions of the component code \(\CW\). %
Let \(A(b_1)\) denote the number of codewords of weight \(b_1\) in \(\CW\).
For \(t = 2, 3\), we calculate the weight distributions of the BCH codes by the MacWilliams identity
\cite[Theorem~3.6]{MoonBook} %
from the distributions of the corresponding dual codes, given in \cite[Sec.~6.1.3]{MoonBook}. %
For BCH codes with an unknown weight distribution, we use the asymptotically-tight binomial approximation
\[
    A(b_1) \approx
    \begin{cases}
        2^{-\nu t} \binom{n}{b_1} & \text{if \(2t + 1 \leq b_1 \leq n - (2t + 1)\)} \\
        1 & \text{if \(b_1 = 0, b_1 = n\)} \\
        0 & \text{otherwise},
    \end{cases}
\]
where \(n = 2^{\nu} - 1\) \cite[Eq.~(17)]{jianApproachingCapacity2017}. For large \(n\), there exists a bound on the relative error of the approximation of order \(n^{-0.1}\) \cite{sidel1971weight}.
The weight distribution \(A_{\textsub{Ev}}(b_1)\) of the even-weight subcode of a BCH code with weight distribution \(A(b_1)\) is \(A_{\textsub{Ev}}(b_1) = A(b_1)\) if \(b_1\) is even and \(A_{\textsub{Ev}}(b_1) = 0\) otherwise.
The weight distribution \(A_{\textsub{Sh}}(b_1)\) of a shortened code based on an BCH code or even-weight subcode of weight distribution \(A(b_1)\) and length \(n + 1\) is
\[
    A_{\textsub{Sh}}(b_1) = \frac{n + 1 - b_1}{n + 1} A(b_1).
\]
This follows directly from Theorem \ref{thm:Fixed_Weight_Dis} below because BCH codes and their even-weight subcodes are cyclic.

Besides the weight distribution, the biweight distribution \cite[Ch.~5. §6]{MacWilliamsSloane} %
is required.%
\footnote{In \cite{MacWilliamsSloane}, the biweight distribution is called ``biweight enumerator''.}
Its coefficients \(B(b_{11}, b_{10}, b_{01}, b_{00})\) count the number of ordered codeword pairs \((\vec{c}_1, \vec{c}_2) \in \CW^2\) that have the configuration \((b_{11}, b_{10}, b_{01}, b_{00})\), which measures the overlapping symbols of \(\vec{c}_1\) and \(\vec{c}_2\): An ordered pair \((\vec{c}_1, \vec{c}_2)\) has the configuration \((b_{11}, b_{10}, b_{01}, b_{00})\) if
\(
    b_{fg} = 
    |\{i \in \{1, \dotsc, n\} : c_{1, i} = f, c_{2, i} = g\}|
\)
holds for all \(f, g \in \{0, 1\}\). For instance, a pair has the configuration \((1, 0, n-1, 0)\) if at one positions both \(\vec{c}_1\) and \(\vec{c}_2\) have a \(1\) and at the other ones \(\vec{c}_1\) has a \(0\) and \(\vec{c}_2\) a \(1\).
Obviously, we have \(b_{11} + b_{10} + b_{01} + b_{00} = n\) and \(B(b_{11}, 0, 0, b_{00}) = A(b_{11})\).
To the best of our knowledge, the biweight distribution of BCH codes is not known, however, for our use case, the approximation described in Appendix~\ref{sec:App_Biweight_Dis} yields good results.

In the following calculations, the symbol at position \(k \in \{1, \dotsc, n\}\) of a codeword is often fixed. In this case, \(A_k^\alpha(b_1)\) denotes the number of codewords \(\vec{c}_1 \in \CW\) of weight \(b_1\) with \(c_{1, k} = \alpha\) and 
\(B_k^{\alpha \beta}(b_{11}, b_{10}, b_{01}, b_{00})\) is the biweight distribution with \(c_{1, k} = \alpha\) and \(c_{2, k} = \beta\) (\(\alpha, \beta \in \{0, 1\}\)).
For cyclic codes (e.g. BCH codes or their even-weight subcodes), we have the following theorem.
\begin{theorem}
\label{thm:Fixed_Weight_Dis}
    For a cyclic code of length \(n\), we have
    \begin{align}
        \label{eqn:Fixed_Weight_Dis}
        A_k^{\alpha}(b_1) 
        &= 
        \frac{b_\alpha}{n} A(b_1), \\
        \label{eqn:Fixed_Biweight_Dis}
        B_k^{\alpha \beta}(b_{11}, b_{10}, b_{01}, b_{00}) 
        &= 
        \frac{b_{\alpha \beta}}{n}
        B(b_{11}, b_{10}, b_{01}, b_{00}).
    \end{align}
\end{theorem}
\begin{proof} 
    See Appendix~\ref{proof:Fixed_Weight_Dis}.
\end{proof}
For shortened codes, which are not, in general, cyclic, we use \eqref{eqn:Fixed_Weight_Dis} and \eqref{eqn:Fixed_Biweight_Dis} as an approximation for \(A_k^{\alpha}\) and \(B_k^{\alpha \beta}\).

\subsubsection{EaED+}\label{sec:Trans_Prob_EaEDPlus}
We now derive \(\T{\alpha}{\beta}{\Dp, \Ep}\) for the EaED+ %
 based on \cite{jianApproachingCapacity2017} and \cite[Sec.~3.7.2]{MoonBook}. %
Consider a random experiment in which an error pattern \(\vec{e}\) is chosen from
\[
    \Omega \coloneqq
    \{
        \vec{e} \in \ZQO^n : %
        \text{%
            \(e_k = \alpha\) and
            \(\Error(\Dp, \Ep)\)%
        }
    \}
\]
uniformly at random. 
Let \(\Mset \subset \Omega\) be the subset that contains only the error patterns \(\vec{e}\) whose decoding result \(\vec{w} \coloneqq \DF_{\textsub{EaED+}}(\vec{e})\) fulfills \(w_k = \beta\). 
Then, the transition probability can be calculated through
\(
    \T{\alpha}{\beta}{\Dp, \Ep} = |\Mset| / |\Omega|
\),
where \(|\Omega| = \binom{n - 1}{\Dp, \Ep}\).
Because of \(\alpha \neq \beta\), \(\Mset\) contains exactly these error patterns of \(\Omega\) that are in \(\Sp_t^3(\vec{c})\) of a codeword 
\(
    \vec{c} \in \CW_k^\beta
    \coloneqq \{\vec{c} \in \CW : c_k = \beta\}
\).%

To count these error patterns, we consider a codeword \(\vec{c} \in \CW_k^\beta\) and an error pattern \(\vec{e} \in \Omega\), as shown in Fig.~\ref{fig:EaED_Plus_Derivation}. 
For both, the symbol at position \(k\) is fixed: \(c_k = \beta\) and \(e_k = \alpha\). 
\begin{figure}
    \centering
    \includegraphics{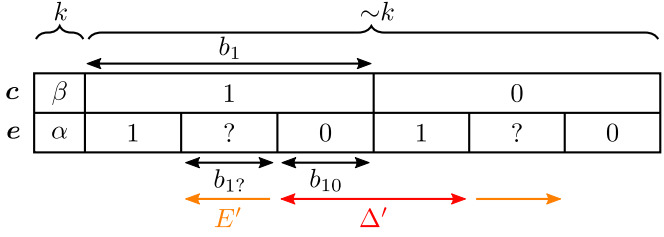}
    \caption{Schematic illustration of the variables in the derivation of the transition probabilities of the EaED+: The symbols of \(\vec{e}\) at \(\compl{k}\) are divided into groups at the \(1\)s of \(\vec{c}\) (1-coordinates) and into groups at the \(0\)s (0-coordinates).}
    \label{fig:EaED_Plus_Derivation}
\end{figure}
At the remaining positions \(\compl{k}\), \(\vec{e}\) has \(\Ep\) erased positions. In addition, let \(\Delp\) of the unerased positions differ from \(\vec{c}\). We call these positions ``differences''.
For \(\vec{e} \in \Sp_t^3(\vec{c})\), \(\Delp\) must be in the range of
\[
    0 \leq \Delp \leq \Delp_{\textsub{max}} 
    \coloneqq
    \left\lfloor
    \frac{\ddesign(t) - \Ep - 1 - \ind{\alpha=\que}}{2}
    \right\rfloor
    - \ind{\alpha \neq \que}.
\]

Moreover, let \(\vec{e}\) have \(b_{1\que}\) erasures and \(b_{10}\) differences at the \(1\)-coordinates of \(\vec{c}\) and the remaining \(\Ep - b_{1\que}\) erasures and \(\Delp - b_{10}\) differences at the \(0\)-coordinates.
Then, since \(\vec{e}\) must have \(\Dp\) \(1\)s at \(\compl{k}\), the weight \(b_1\) of \(\vec{c}\) at \(\compl{k}\) must be
\[
    b_1 = b_1(\Delp, b_{1\que}, b_{10}) \coloneqq \Dp - \Delp + b_{1\que} + 2 b_{10}.
\]
There are \(A_k^\beta(b_1 + \ind{\beta=1})\) codewords of \(\CW_k^\beta\) of weight \(b_1\). For each codeword, there are
\[
    \Theta(\Delp, b_{1\que}, b_{10}, b_1) \coloneqq \\
    \binom{b_1}{b_{1\que}, b_{10}}
    \binom{n - 1 - b_1}{\Delp - b_{10}, \Ep - b_{1\que}},
\]
different error patterns \(\vec{e}\) whose erasures and differences at \(\compl{k}\) are distributed as defined above by \(\Delp\), \(b_{1\que}\) and \(b_{10}\).

By summing over all possible combinations of \(\Delp\), \(b_{1\que}\) and \(b_{10}\), we obtain
\begin{multline*}
    |\Mset|
    =
    \sum_{\Delp = 0}^{\Delp_{\textsub{max}}}
    \sum_{b_{10} = 0}^{\Delp}
    \sum_{b_{1\que} = 0}^{\Ep} \Big(
    A_k^\beta\big(b_1(\Delp, b_{1\que}, b_{10}) + \ind{\beta=1}\big)\\
    \Theta\big(\Delp, b_{1\que}, b_{10}, b_1(\Delp, b_{1\que}, b_{10})\big)
    \Big)
    ,
\end{multline*}
where we use the convention that \(A_k^\beta(b_1) = 0\) if \(b_1 < 0\) or \(b_1 > n\) and \(\binom{n}{k_1, k_2} = 0\) if \(n < 0\), \(k_1 < 0\), \(k_2 < 0\) or \(k_1 + k_2 > n\).
Note that no error pattern is counted twice as all spheres \(\Sp_t^3(\vec{c})\) are disjoint, which is an implication of Theorem~\ref{thm:EaED_correcting_capability}.

\subsubsection{EaED}
The derivation of \(\T{\alpha}{\beta}{\Dp, \Ep}\) of the EaED is based on the same principle as the derivation for the EaED+ above. It is described in Appendix~\ref{sec:Decoder_Trans_Probs_EaED}.

\subsection{Noise Threshold}
We use the DE recursion of \(\chiM^{(\ell)}\) to evaluate the performance of the code over the channel \((\lEsNO, T)\). 
\subsubsection{GLDPC Ensemble}
We first focus on the GLDPC ensemble.
First, the channel transition probabilities \(\chiC\) of \((\lEsNO, T)\) are calculated via \eqref{eqn:Channel_Trans_Probs}.
Then, the recursion \(\chi_{\textsub{rec}}\) is applied \(\ell\) times to \(\chiM^{(1)} = \chiC\) resulting in 
\(
    \chiM^{(\ell + 1)} 
    = (\deltaM^{(\ell + 1)}, \epsilonM^{(\ell + 1)})
\).
The bit error probability after \(\ell\) decoding iterations is given by
\begin{equation}\label{eqn:Def_BER}
    \phantom{.}
    \BEP^{(\ell)}\left(\lEsNO, T\right) \coloneqq
    \deltaM^{(\ell + 1)} + \frac{1}{2} \epsilonM^{(\ell + 1)}
    .
\end{equation}
\(\BEP^{(\ell)}\) is used to define the noise threshold
\begin{equation} \label{eqn:noise_threshold}
    \left(\EsNO\right)^{\mathclap{\ast}}(T) 
    \coloneqq 
    \inf
    \bigg\{
        \EsNO \geq 0:
        \lim_{\ell \to \infty}
        \BEP^{(\ell)}\left(\EsNO, T\right)
        = 0
    \bigg\}
\end{equation}
as a performance measure of the channel~\cite{RichardsonCapa}. %
For this definition, we assume that \(\BEP^{(\ell)}\left(\lEsNO, T\right)\) is a monotonically decreasing function in \(\lEsNO\).

\subsubsection{SC-GLDPC Ensemble}
For the SC-GLDPC ensemble only the first \(32\) groups of VNs are considered, to keep the computational effort of the DE manageable.
Their error and erasure probabilities are initialized with 
\(\vec{\chi}_{\textsub{m}, 0}^{(1)} = (0, 0)\) and \(\vec{\chi}_{\textsub{m}, i}^{(1)} = \chiC\) for \(i > 0\). Then, recursion~\eqref{eqn:SC_GLPDC_Rec} is applied \(\ell\) times to \(\vec{\chi}_{\textsub{m}, i}^{(1)}\) resulting in \(\vec{\chi}_{\textsub{m}, i}^{(\ell + 1)}\). To calculate \(\BEP^{(\ell)}\), the error and erasure probabilities of VN group \(i = 1\) to \(10\) are averaged to \(
    \chiM^{(\ell + 1)} 
    = (\deltaM^{(\ell + 1)}, \epsilonM^{(\ell + 1)})
\).
Then, \(\BEP^{(\ell)}\) and the noise threshold are determined by \eqref{eqn:Def_BER} and \eqref{eqn:noise_threshold} with \(\chiM^{(\ell + 1)}\).

We limit the calculation of \(\BEP^{(\ell)}\) on the first \(10\) groups to reduce the computational effort. If the bit error probability of the first \(10\) groups converges to \(0\), it can be assumed that the bit error probability of the following groups will also converge to \(0\). Furthermore, this limitation justifies the consideration of only the first \(32\) groups in the DE. The following groups would only have a negligible effect on the performance of the first \(10\) groups because they are too far away.

\subsubsection{Numerical Estimation}
For the numerical estimation of the limit in \eqref{eqn:noise_threshold}, the recursion is applied until the change of \(\BEP^{(\ell)}\) in one iteration is less than \num{e-12}. The infimum of the set in \eqref{eqn:noise_threshold} is calculated by a binary search, which searches for the minimal \(\lEsNO\) with \(\lim_{\ell \to \infty} \BEP^{(\ell)}(\lEsNO, T) < \num{e-10}\).
We use \(\num{e-10}\) to avoid numerical instabilities, which occurred for lower error probabilities.

\section{Results}\label{sec:results}
\subsection{Theoretical Results}
We evaluated the noise threshold \((\lEsNO)^{\ast}(T)\) numerically for different \(T\) using the DE analysis based on either the GLDPC or the SC-GLDPC ensemble. The result of a product code (GLDPC ensemble) of a \((511, 484, 3)\)-BCH code is shown in Fig.~\ref{plot:Noise_Thr_DE}. 
\begin{figure}
    \centering
    \includegraphics{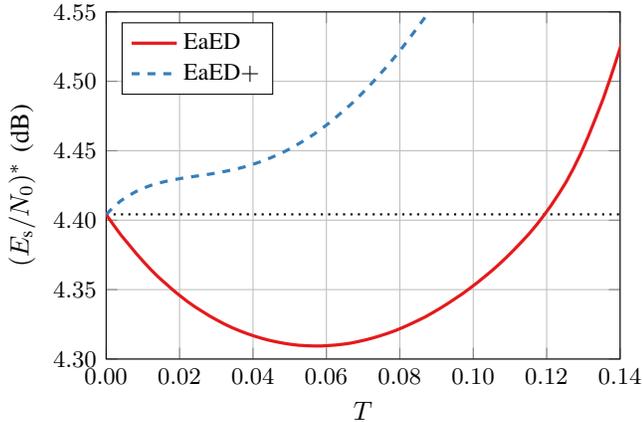}
    \caption{Noise thresholds calculated via DE for the \((511, 484, 3)\)-BCH product code. The dotted line marks the noise threshold of HDD.}
    \label{plot:Noise_Thr_DE}
\end{figure}
The dotted line marks the performance of HDD (\(T=0\)) for the EaED and EaED+.

The threshold of the EaED has a minimum at \(T \neq 0\), i.e., EaED performs better than HDD. To quantify the performance increase of the EaED compared to HDD, we define the optimal \(T\) by
\begin{align}
    \Topt &\coloneqq \argmin_{T \geq 0}\left\{ (\lEsNO)^{\ast}(T) \right\} \label{eqn:topt}\\
    \intertext{and the decrease in \((\lEsNO)^{\ast}\) at \(\Topt\) compared to HDD by the predicted gain}
    \gainTheo &\coloneqq (\lEsNO)^{\ast}(0) - (\lEsNO)^{\ast}(\Topt).\label{eqn:gainTheo}
\end{align}
For this code, we get for the EaED performance:
\(\Topt = \num[round-mode = places, round-precision=3]{0.056935}\)
and \(\gainTheo = \SI[round-mode = places, round-precision=3]{0.094738}{\dB}\).

However, the EaED+ has its minimum noise threshold at \(T=0\). For it, the use of error-and-erasure decoding results in a worse performance and erasures are not beneficial. %
One explanation for this behavior is as follows: The errors and erasures of a component code can be corrected by the EaED+ if \(2 D + E < \ddesign(t)\) is fulfilled. For \(T > 0\), because of AWGN, more correctly than incorrectly received bits are mapped to erasures. 
Hence, on average, \(2 D + E\) could be larger than \(2 D\) for \(T=0\), which results in a performance decrease. The EaED, on the other hand, can also correct some error patterns outside these spheres.

For larger values of \(T\), the noise threshold increases significantly for both decoders. The reason for the increase is that for large \(T\), many correctly received symbols are mapped to erasures, which results in a loss of information.

\subsubsection{Parameter Analysis BCH Code}
\begin{figure*}[tb]
    \captionsetup[subfigure]{justification=justified,singlelinecheck=false}
    \centering
    \subfloat[Results of the product and staircase codes of a usual and shortened BCH code, respectively.\label{plot:parameter_not_even}]{%
        \includegraphics{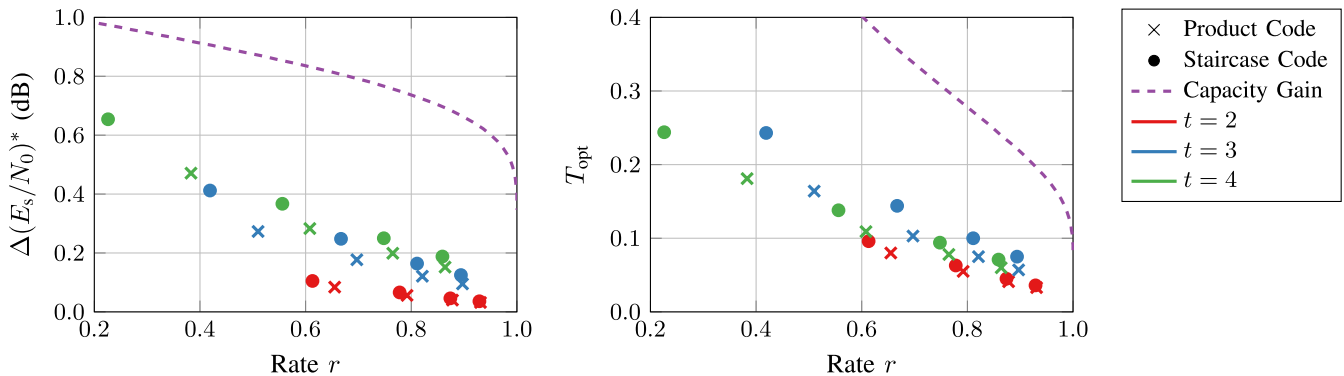}
    }\\
    \subfloat[Results of the product codes of a usual BCH code or even-weight subcode.\label{plot:parameter_even}]{%
        \includegraphics{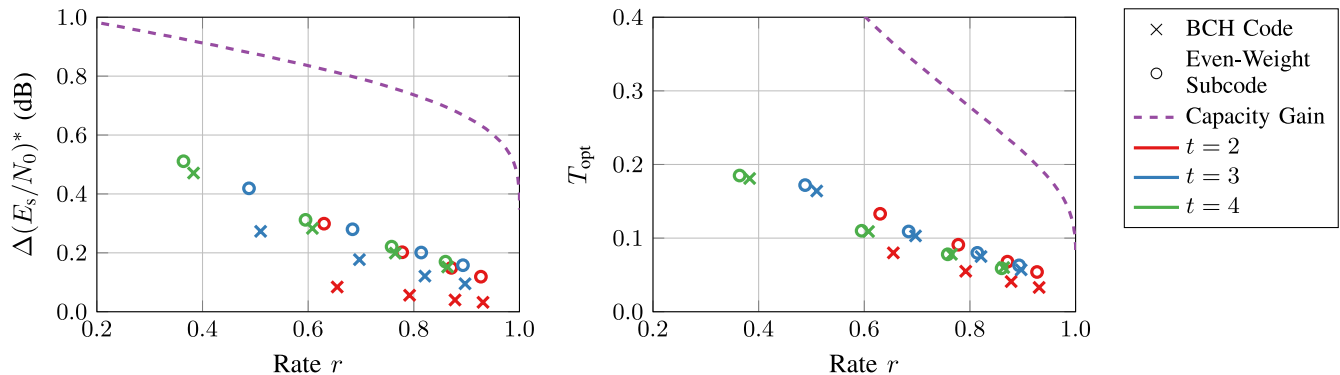}
    }
    \caption{Results of the parameter analysis: The predicted noise threshold gain \(\gainTheo\) that the EaED achieves compared to HDD is plotted on the left and the corresponding \(\Topt\) on the right.
    The component code is an \((n, k, t)\)-BCH code with \(n \in \{63, 127, 255, 511\}\) and \(t \in \{2, 3, 4\}\) or its even-weight subcode or shortened code.
    The dashed curves are the results of the capacity analysis. They mark the maximal achievable predicted gain in theory and the corresponding \(\Topt\) (Sec.~\ref{sec:channel}).
    }
    \label{plot:parameter}
\end{figure*}

We now analyze the predicted gain \(\gainTheo\) for different component codes. We limit this analysis to the EaED as this decoder is the most relevant in practice.
Figure~\ref{plot:parameter}-\subref{plot:parameter_not_even} shows \(\gainTheo\) for several product and staircase codes (SC-GLDPC ensemble) of a BCH and shortened BCH code, respectively, plotted as a function of their rates. 
The corresponding \(\Topt\) %
is shown on the right of Fig.~\ref{plot:parameter}-\subref{plot:parameter_not_even}.
The dashed curves in Fig.~\ref{plot:parameter} are the result of the capacity analysis
(Sec.~\ref{sec:channel}) and show the capacity gain, i.e. the maximal predicted gain that could be expected if error-and-erasure decoding is used instead of~HDD. %

The predicted gain increases with decreasing length \(n\) of the BCH code (decreasing rate in the diagram). For instance, the predicted gain of the product and staircase codes of BCH codes with \(n = 511\) increases from less than \SI{0.19}{\dB} to \SI{0.65}{\dB} for the staircase code of a \((63, 39, 4)\)-BCH code.
A possible reason may be that, for fixed \(t\), the number of correctable erasures per bit decreases with \(n\) according to Theorem~\ref{thm:EaED_correcting_capability}.
Furthermore, the predicted gain increases with \(t\), and all staircase codes achieve a larger predicted gain than the product code of the same component code.

\subsubsection{Parameter Analysis Even-Weight Subcode}
Figure~\ref{plot:parameter}-\subref{plot:parameter_even} shows the predictd gain \(\gainTheo\) and \(\Topt\) of product codes that are constructed from an even-weight subcode (circles) compared to the results of Fig.~\ref{plot:parameter}-\subref{plot:parameter_not_even} (crosses).
For the sake of clarity, the results of the staircase codes are omitted as they are similar to the ones of the product codes.

The use of the even-weight subcode leads to an increase in the predicted gain, in particular for smaller values of \(t\).
This increase can be motivated using Theorem~\ref{thm:EaED_correcting_capability}: a word is corrected for \(2 D + E \leq 2 t\) if a BCH code is used and for \(2 D + E \leq 2 t + 1\) if its even-weight subcode is used.
Hence, using an even-weight subcode enables the correction of one extra erasure.
This explains why even-weight subcodes benefit more from error-and-erasure coding. 
Furthermore, it explains the large increase for \(t=2\): Because of the small error-correcting capability, the extra erasure has a greater impact than for larger~\(t\).

\subsection{Simulation}
To check if the theoretical results of the DE are consistent with the performance, we simulated the performance of product and staircase codes.
In this section, we define the noise threshold \((\lEsNO)^{\ast\ast}\) as that \(\lEsNO\) for which the output \(\BER\) is equal to \(\BER_{\textsub{target}} \coloneqq \num{e-4}\) after \num{20} decoding iterations.
The simulated gain and \(\Topt\) are defined in the same way as the predicted gain and \(\Topt\) in \eqref{eqn:gainTheo} and \eqref{eqn:topt}.

In the simulation, the points of the \(\BER\)-\(\lEsNO\)-curve are estimated by a Monte Carlo method, along with a binary search to determine the intersection of the curve with \(\BER_{\textsub{target}}\) at \((\lEsNO)^{\ast\ast}\).
During the binary search, the number of trials is dynamically adapted to ensure that despite the randomness of the simulations, the estimated \(\BER\) is greater or smaller than \(\BER_{\textsub{target}}\) with sufficiently large confidence.

Figure~\ref{plot:simulation} compares the simulation results of a product code of the \((511, 484, 3)\)-BCH code with the results of the DE analysis of Fig.~\ref{plot:Noise_Thr_DE}.
For each decoder, we did two simulations, one with EMP and one with IMP.
\begin{figure*}[tb]
    \centering
    \subfloat[EaED]{%
        \includegraphics{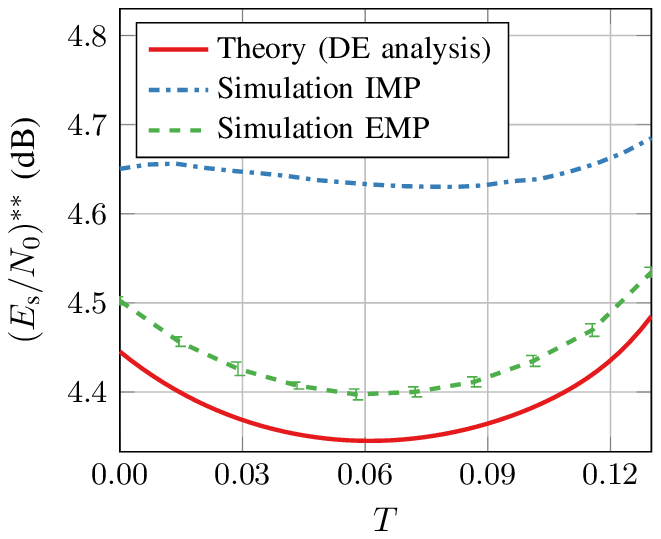}%
    }%
    \qquad
    \subfloat[EaED+]{%
        \includegraphics{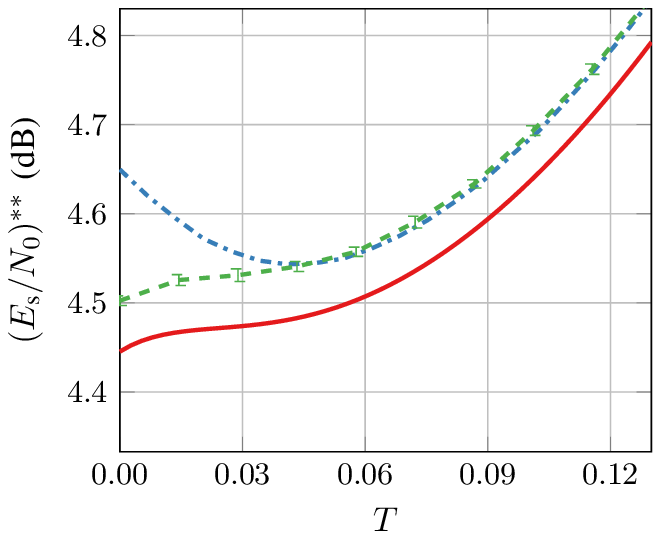}%
    }%
    \caption{
        Simulation results of a product code of the \((511, 484, 3)\)-BCH code compared with the results of the DE analysis.
        The error bars of the EMP curve are the remaining search interval after the termination of the binary search. The error bars of the IMP results were omitted because they are negligible.
    }
    \label{plot:simulation}
\end{figure*}
The plots show an approximately constant gap between the predicted thresholds and the simulated \((\lEsNO)^{\ast\ast}\) with EMP decoding (The gap slightly decreases over \(T\) and is in the range of \SIrange{0.050}{0.058}{\dB}.). 
The gap is due to finite length effects because the DE analysis considers GLDPC graphs of infinite size in contrast to the finite size of the simulated product code.
Since the gap is approximately constant over \(T\), the predicted gain and \(\Topt\) of the DE analysis match those of the results in practice.

However, for both decoders, the curve ``Simulation IMP'' has no similarity to the theory.
Hence, an estimation of the simulated gain of error-and-erasure coding with DE is not possible if IMP is used.
Nevertheless, the IMP performance of the EaED+ is quite surprising:
Although this decoder achieves no simulated gain using EMP decoding, it achieves a simulated gain of around \SI{0.106}{\dB} at \(\Topt = \num{0.04}\) using IMP decoding.
It outperforms IMP decoding of the EaED, which has only a negligible simulated gain.

Furthermore, we simulated the product code of the \((63, 45, 3)\)-BCH code that is decoded by the EaED (results not shown). 
In this case, the DE analysis underestimates the simulated gain of the EMP simulation by \SI{21}{\percent}, while \(\Topt\) is calculated correctly.
The difference between predicted and simulated gain may result from finite length effects and the approximation of the biweight distribution.%

Figure \ref{plot:ber_curves_simulation} shows the simulated BER curves of a product code of the \((511, 484, 3)\)-BCH code that is decoded by both EaED and EaED+ with \(20\) decoding iterations using either EMP or IMP. For \(T\), we choose \(0\) or the \(\Topt\) of the respective decoder.
We observe that error-and-erasure decoding does not lead to early error floors and that the gains are consistent with the DE results.

\begin{figure}[tb]
    \centering
    \includegraphics{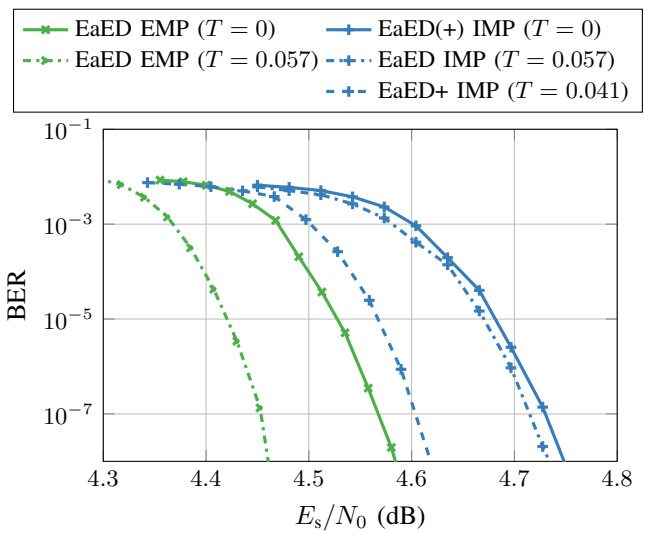}
    \caption{Simulated BER curves for a product code of the \((511, 484, 3)\)-BCH code using EMP and IMP decoding, respectively.
    \label{plot:ber_curves_simulation}}
\end{figure}

\section{Conclusions \& Outlook}\label{sec:conclusion}
We analyzed the error-and-erasure decoding of product and staircase codes based on BCH codes or their even-weight subcodes. 
For the analysis, we formulated DE on the corresponding GLDPC or SC-GLDPC ensembles that are decoded with EMP.
We have shown that error-and-erasure decoding archives a gain in \(\lEsNO\) compared to HDD, whereby the predicted gain is larger for lower rate codes and if an even-weight subcode is used as a component code. Finally, we have verified the results by a simulation of a product code using both EMP decoding but also the simpler IMP decoding, where we also observed predicted gains for a variation of the component code decoders.

In practice, instead of using the even-weight subcodes as component codes, BCH codes are often extended by a parity check bit. 
Since these codes have also an even design distance, we assume that their predicted gains are comparable with the results of the even-weight subcodes.
A detailed analysis of extended BCH codes is subject of further work.

\appendices
\section{Proof of Theorem~\ref{thm:EaED_correcting_capability}}
\label{proof:EaED_correcting_capability}
\begin{proof}
    Based on \cite[Sec.~3.8.1]{MoonBook}: For \(\vec{y}_1\), the EaED assigns to the erased coordinates of \(\vec{y}\) the random vector \(\vec{p}\) and for \(\vec{y}_2\), the inverted vector \(\xor{\vec{p}}\).
    Because of this assignment, \(\vec{y}_1\) has \(D_1 \leq E\) errors in addition to the \(D\) errors of \(\vec{y}\).
    At the erased coordinates, \(\vec{y}_2\) has errors where \(\vec{y}_1\) has no errors because, for \(\vec{y}_2\), the inverted vector \(\xor{\vec{p}}\) is inserted. Hence, \(\vec{y}_2\) has  \(D_2 = E - D_1\) errors besides the \(D\) errors.
    Therefore, \(D_i \leq E / 2\) holds for at least one \(\vec{y}_i\) with \(i \in \{1, 2\}\). The total number of errors of this \(\vec{y}_i\) fulfills
    \[
        \phantom{.}
        D + D_i \leq D + \frac{E}{2} 
        \stackrel{\text{(a)}}{<}%
        t + 1
        \Rightarrow
        D + D_i \leq t
        ,
    \]
    where (a) holds because of \eqref{eqn:Condition_EaED_correcting_capability} and \(\ddesign(t) \leq 2 t + 2\) for the defined component codes.
    Hence, the BDD decodes at least one \(\vec{y}_i\) to the right codeword.
    
    It remains to prove that if both results \(\vec{w}_1\) and \(\vec{w}_2\) are codewords, the EaED selects the correct codeword. A wrong selection is only possible if one decoding result is not correct. Let \(\vec{w}_c\) be the correct and \(\vec{w}_e\) the erroneous result of \(\vec{w}_1\) and \(\vec{w}_2\). Suppose that \(\vec{w}_e\) is falsely selected. Then the following inequality contradicts \eqref{eqn:Condition_EaED_correcting_capability} (\(\dnE{\vec{y}}\) and \(\dE{\vec{y}}\) are the distance at the unerased and erased coordinates of \(\vec{y}\), respectively.):
    \begin{align*}
        &\ddesign(t)
        \leq
        \dmin
        \leq 
        \dH(\vec{w}_e, \vec{w}_c)  \\
        &= \dnE{\vec{y}}(\vec{w}_e, \vec{w}_c) + \dE{\vec{y}}(\vec{w}_e, \vec{w}_c)
        \smash{\stackrel{\text{(a)}}{\leq}}
        \dnE{\vec{y}}(\vec{w}_e, \vec{w}_c) + E \\
        &\stackrel{\text{(b)}}{\leq}
        \dnE{\vec{y}}(\vec{w}_e, \vec{y}) + \dnE{\vec{y}}(\vec{y}, \vec{w}_c) + E
        \stackrel{\text{(c)}}{\leq} 
        2 D + E\,,
    \end{align*}
    where (a) holds because the distance of two words of \(E\) coordinates is at most \(E\).
    (b) is the triangle inequality and (c) uses that \(\dnE{\vec{y}}(\vec{y}, \vec{w}_c) = D\) because \(\vec{y}\) has \(D\) errors at the unerased coordinates. Moreover, according to the assumption, \(\dnE{\vec{y}}(\vec{w}_e, \vec{y}) \leq \dnE{\vec{y}}(\vec{y}, \vec{w}_c) = D\), as otherwise \(\vec{w}_e\) would not have been selected.
\end{proof}

\section{Proof of Theorem~\ref{thm:Performance_CW_Independent}}
\label{proof:Performance_CW_Independent}
\begin{proof}
Let \(\oplus\colon\,\ZO^n \times \ZQO^n \to \ZQO^n\) be an operator that computes for each component
\[
    [\vec{a} \oplus \vec{b}]_i \coloneqq
    \begin{cases}
        a_i + b_i & \text{if \(b_i \neq \que\)} \\
        \que & \text{otherwise}.
    \end{cases}
\]
Then, it is easy to see that the BDD and the EaED+ fulfill the symmetry condition
\begin{equation}
\label{eqn:Symmetrie_Condition}
    \phantom{,}
    \DF(\vec{c} \oplus \vec{e}) = \vec{c} \oplus \DF(\vec{e})
    \quad
    \text{for all \(\vec{c} \in \CW\) and \(\vec{e}\)},
\end{equation}
where, in the case of the BDD, we define \(\vec{c} \oplus \text{fail} := \text{fail}\) and \(\vec{e}\) is an error-and-erasure pattern (\(\vec{e} \in \ZQO^n\)).

For the EaED, we interpret words of \(\ZQO^n\) as random variables taking values in \(\ZQO^n\), so that \(\DF_{\textsub{EaED}}\) is a function that transforms random variables.
Then, the EaED fulfills the symmetry condition
\begin{equation}
\label{eqn:Symmetrie_Condition_EaED}
    \phantom{,}
    \DF_{\textsub{EaED}}(\vec{c} \oplus \vec{e}) \mathrel{\overset{d}{=}} \vec{c} \oplus \DF_{\textsub{EaED}}(\vec{e})
    \quad
    \text{for all \(\vec{c} \in \CW\) and \(\vec{e}\)},
\end{equation}
where \(\vec{e}\) is an arbitrary random variable on \(\ZQO^n\) and ``\(\mathrel{\overset{d}{=}}\)'' means that the random variables are equal in distribution.

To prove this condition, we require an alternative description of the EaED. Let \(\Lambda(\vec{w}_1, \vec{w}_2, \vec{y})\) be the function that determines the decoding result of \(\vec{y}\) from \(\vec{w}_1\) and \(\vec{w}_2\) in decoding step~\ref{item:Decision_EaED} of the EaED (Sec.~\ref{sec:Decoder}). 
For the sake of clarity, we decompose \(\vec{y}\) into an unerased and and erased component: \(\vec{y} = [\vec{y}_\nE, \vec{\que}]\). Using \(\Lambda\), we get
\begin{equation*}
    \DF_{\textsub{EaED}}(\vec{y}) 
    = 
    \Lambda(\DF_{\textsub{BDD}}([\vec{y}_\nE, \vec{p}]), \DF_{\textsub{BDD}}([\vec{y}_\nE, \xor{\vec{p}}]), \vec{y}),
\end{equation*}
where the erased coordinates of \(\vec{y}\) are replaced by \(\vec{p}\), which is a uniform random variable on \(\ZO^n\). 
Let \(\vec{c}\in\CW\) be a codeword and \(\vec{e}\) be an arbitrary random variable on \(\ZQO^n\). We decompose \(\vec{c}\) and \(\vec{e}\) into the bits at the unerased and erased coordinates of \(\vec{e}\) giving
\(\vec{c} = [\vec{c}_\nE, \vec{c}_{\E}]\) and \(\vec{e} = [\vec{e}_\nE, \vec{\que}]\). By doing so, we get
\begin{align*}
    \label{eqn:symmetry_cond_proof}
    &\DF_{\textsub{EaED}}(\vec{c} \oplus \vec{e})
    =
    \DF_{\textsub{EaED}}([\vec{c}_\nE + \vec{e}_\nE, \vec{?}]) \nonumber \\
    &=
    \Lambda\left(
        \DF_{\textsub{BDD}}([\vec{c}_\nE + \vec{e}_\nE, \vec{p}]), 
        \DF_{\textsub{BDD}}([\vec{c}_\nE + \vec{e}_\nE, \xor{\vec{p}}]), 
        \vec{c} \oplus \vec{e}\right) \nonumber \\
    &\stackrel{\text{(a)}}{=}
    \begin{aligned}[t]
    \Lambda(
        &\vec{c} \oplus \DF_{\textsub{BDD}}([\vec{e}_\nE, \vec{c}_{\E} + \vec{p}]),\\
        &\vec{c} \oplus \DF_{\textsub{BDD}}([\vec{e}_\nE, \vec{c}_{\E} + \xor{\vec{p}}]), 
        \vec{c} \oplus \vec{e})
    \end{aligned}\\
    &=
    \vec{c} \oplus
    \Lambda\left(
        \DF_{\textsub{BDD}}([\vec{e}_\nE, \vec{c}_{\E} + \vec{p}]), 
        \DF_{\textsub{BDD}}([\vec{e}_\nE, \vec{c}_{\E} + \xor{\vec{p}}]), 
        \vec{e}\right), \nonumber
\end{align*}
where (a) uses the symmetry condition of the BDD.
Finally, we substitute \(\vec{p}_2 \coloneqq \vec{c}_E + \vec{p}\) and \(\xor{\vec{p}}_2 = \vec{c}_E + \xor{\vec{p}}\) resulting in
\begin{align*}
    \DF_{\textsub{EaED}}(\vec{c} \oplus \vec{e})
    &=
    \vec{c} \oplus
    \Lambda\left(
        \DF_{\textsub{BDD}}([\vec{e}_\nE, \vec{p}_2]), 
        \DF_{\textsub{BDD}}([\vec{e}_\nE, \xor{\vec{p}}_2]), 
        \vec{e}\right) \\
    &\mathrel{\overset{d}{=}} \vec{c} \oplus \DF_{\textsub{EaED}}(\vec{e}),
\end{align*}
because \(\vec{p}_2\), just as \(\vec{p}\), is a uniform random variable on \(\ZO^n\).
This proves the symmetry condition of the EaED.

Using the respective symmetry condition \eqref{eqn:Symmetrie_Condition} or \eqref{eqn:Symmetrie_Condition_EaED}, it can be shown, similar to \cite{RichardsonCapa}, that the expected number of errors and erasures of the whole GLDPC decoder is independent of the transmitted codeword.
\end{proof}

\section{Alternative description of EMP}
\label{sec:DE_EMP}
It is not immediately obvious that DE is allowed for the proposed message passing algorithm, as the channel input values are used in the CN update. Therefore, we present an alternative description of the same message passing algorithm in which the decoding of the component codes is moved from the CN to the VN update. This allows the insertion of the channel input value at the VN similar to the approach in \cite{jianApproachingCapacity2017}.

The message passing starts with the initialization \(\nu_{i, j}^{(1)} = \nu_{i, j^\prime}^{(1)} = r_i\) of the outgoing VN messages. During the CN update, each CN \(j\) combines the incoming messages into a vector. For each VN \(i\), connected at socket \(k=\sigma^{-1}(i) \in \{1, \dotsc, n\}\), it replaces the \(k\)-th symbol by a blank \(\square\) and returns the vector to the VN:
\[
    \phantom{.}
    \tilde{\nu}_{i, j}^{(\ell)} 
    \coloneqq
    (
        \nu_{\sigma_j(1), j}^{(\ell)}, \dotsc ,
        \nu_{\sigma_j(k-1), j}^{(\ell)}, \square, 
        \nu_{\sigma_j(k+1), j}^{(\ell)}, \dotsc
        , \nu_{\sigma_j(n), j}^{(\ell)}
    ).
\]
Due to the replacement of the \(k\)-th message, there is only extrinsic information passed.

In the VN update, each VN \(i\) receives two messages from its connected CNs \(j\), \(j^\prime\). To calculate the outgoing message for CN \(j\), the VN takes the incoming message of the respective other CN \(j^\prime\) and replaces the blank by its own channel input value \(r_i\) resulting in 
\[
    \phantom{.}
    \vec{y}_{j, \textsub{EMP}, k}^{(\ell)} 
    \coloneqq
    (
        \nu_{\sigma_j(1), j}^{(\ell)}, \dotsc ,
        \nu_{\sigma_j(k-1), j}^{(\ell)}, r_{i}, 
        \nu_{\sigma_j(k+1), j}^{(\ell)}, \dotsc
    )
    ,
\]
which was generated in the original algorithm in the CN update. Then \(\vec{y}_{j, \textsub{EMP}, k}^{(\ell)} \) is decoded and the symbol at the position of the blank \(\square\) is sent to CN \(j^\prime\):
\(\nu_{i, j^\prime}^{(\ell+1)} = \big[\DFC(\vec{y}_{j, \textsub{EMP}, k}^{(\ell)})\big]_k\).

It is easy to see that these VN-to-CN messages are identical with the ones of the original message passing algorithm introduced in Sec.~\ref{sec:background}. This proves that DE can be applied on the original message passing algorithm because only extrinsic information is passed in this scheme.

\section{Biweight Distribution Approximation}
\label{sec:App_Biweight_Dis}
In order that a pair \((\vec{c}_1, \vec{c}_2)\) has the configuration \(\vec{b}_\ast \coloneqq (b_{11}, b_{10}, b_{01}, b_{00})\), \(\vec{c}_1\) must have \(b_1 \coloneqq b_{11} + b_{10}\) \(1\)s and \(b_0 \coloneqq b_{01} + b_{00}\) \(0\)s, which is the case for \(A(b_1)\) codewords.  Moreover, the weight of \(\vec{c}_2\) must be \(\we(\vec{c}_2) \coloneqq b_{11} + b_{01}\) and \(\dH(\vec{c}_1, \vec{c}_2) = b_{10} + b_{01}\).
For \(B\), we use the approximation
\begin{align*}
    &B(b_{11}, b_{10}, b_{01}, b_{00}) \approx A(b_1)\\
    &\cdot \begin{cases}
        A(\we(\vec{c}_2)) & 
        \begin{aligned}[c]
            &A(\we(\vec{c}_2)) = 0 \\
            &\text{or \(\we(\vec{c}_2) \in \{0, n\}\)}
        \end{aligned} \\
        A(b_{11} + b_{01})
        \binom{b_1}{b_{11}} \binom{b_0}{b_{01}} 
        / \binom{n}{b_{11} + b_{01}} & 
        \we(\vec{c}_2) \leq \dH(\vec{c}_1, \vec{c}_2) \\
        A(b_{10} + b_{01}) 
        \binom{b_1}{b_{10}} \binom{b_0}{b_{01}} 
        / \binom{n}{b_{10} + b_{01}} &
        \we(\vec{c}_2) > \dH(\vec{c}_1, \vec{c}_2),
    \end{cases}
\end{align*}
which will be motivated in the following.

In the first case, either no valid pair exists, or \(\vec{c}_2\) is the all-zero or the all-one codeword (if existing). The all-zero or all-one codeword form together with all codewords of weight \(b_1\) a pair of the configuration \(\vec{b}_\ast\). In these cases, no approximation is necessary, and we have \(A(b_1) A(\we(\vec{c}_2))\) valid pairs.

For the second case, we first consider a fixed codeword \(\vec{c}_1\) from the \(A(b_1)\) codewords of weight \(b_1\). Now, we approximate the number of codewords \(\vec{c}_2\) that form together with this  \(\vec{c}_1\) a pair of \(\vec{b}_\ast\).
We know that \(A(\we(\vec{c}_2))\) codewords have the correct weight. Since we have no further information on the code, we assume that each of these codewords is independently and uniformly chosen at random from the set of the binary words of length \(n\) and weight \(\we(\vec{c}_2)\). 
Then, the probability that one of these random words has \(b_{11}\) \(1\)s at the \(1\)-coordinates of \(\vec{c}_1\) and \(b_{01}\) \(1\)s at the \(0\)-coordinates is
\(
    P = \binom{b_1}{b_{11}} \binom{b_0}{b_{01}} / \binom{n}{b_{11} + b_{01}}
\).
Hence, on average, \(A(\we(\vec{c}_2)) P\) codewords form together with \(\vec{c}_1\) a pair of the configuration \(\vec{b}_\ast\). Since there are \(A(b_1)\) possible codewords for \(\vec{c}_1\), we have \(A(b_1) A(\we(\vec{c}_2)) P\) pairs in total.

In the third case, we count the pairs that have the configuration
\[
    \vec{\tilde{b}}_\ast \coloneqq
    (\tilde{b}_{11}, \tilde{b}_{10}, \tilde{b}_{01}, \tilde{b}_{00}) =
    (b_{10}, b_{11}, b_{01}, b_{00})
\]
using the second case.
Each pair \((\vec{\tilde{c}}_1, \vec{\tilde{c}}_2)\) of \(\vec{\tilde{b}}_\ast\) can be transformed into a pair \((\vec{c}_1, \vec{c}_2)\) of \(\vec{b}_\ast\) by the bijective transformation
\(
    (\vec{c}_1, \vec{c}_2) = 
    (\vec{\tilde{c}}_1, \vec{\tilde{c}}_1 + \vec{\tilde{c}}_2)
\)
due to linearity.
Therefore, the biweight distribution of \(\vec{\tilde{b}}_\ast\) and \(\vec{b}_\ast\) are equal.

We observed that the results of the third approximation are better than the second one when \(\we(\vec{c}_2)\) is not too small.

\section{Proof of Theorem~\ref{thm:Fixed_Weight_Dis}}
\label{proof:Fixed_Weight_Dis}
\newcommand{\pair}{\mathcal{K}}
\begin{proof}
For the sake of clarity, we use the abbreviation \(\vec{b}_\ast \coloneqq (b_{11}, b_{10}, b_{01}, b_{00})\) in this proof. Let \(\pair(\vec{b}_\ast)\) be the set of ordered codeword pairs 
\((\vec{c}_1, \vec{c}_2) \in \CW^2\) that have the configuration \(\vec{b}_\ast\). Its cardinality is the biweight distribution \(B(\vec{b}_\ast)\).
Let \(\pair_k^{\alpha \beta}(\vec{b}_\ast) \subset \pair(\vec{b}_\ast)\) be the subset whose pairs \((\vec{c}_1, \vec{c}_2)\) additionally fulfill \(c_{1, k} = \alpha\), \(c_{2, k} = \beta\) with \(k \in \{1, \dotsc, n\}\) and \(\alpha, \beta \in \{0, 1\}\).
Its cardinality is 
\(
    B_k^{\alpha \beta}(\vec{b}_\ast) = |\pair_k^{\alpha \beta}(\vec{b}_\ast)|
\).

First, we show that for cyclic codes, \(B_k^{\alpha \beta}(\vec{b}_\ast)\) 
is independent of \(k\). Let \(i, j \in \{1, \dotsc n\}\), \(i \neq j\) be two different positions.
Consider the function
\[
    s\colon \, \pair_i^{\alpha \beta}(\vec{b}_\ast) \to \ZO^n \times \ZO^n
\]
that cyclically shifts each codeword of a pair \((\vec{c}_1, \vec{c}_2)\) so that the \(i\)-th position is shifted to the \(j\)-th position. Since the code is cyclic, the words of the shifted pair \((\vec{c}_1^s, \vec{c}_2^s) \coloneqq s\big((\vec{c}_1, \vec{c}_2)\big)\) are also codewords. In addition, the shift does not change the configuration \(\vec{b}_\ast\) of the pair, and we have \((\vec{c}_1^s, \vec{c}_2^s) \in \pair(\vec{b}_\ast)\).
Furthermore, according to the definition of \(s\), \(c_{1, j}^s = c_{1, i} = \alpha\) and \(c_{2, j}^s = c_{2, i} = \beta\), which implies \((\vec{c}_1^s, \vec{c}_2^s) \in \pair_j^{\alpha \beta}(\vec{b}_\ast)\). Hence, \(s\) is an injective function from \(\pair_i^{\alpha \beta}(\vec{b}_\ast)\) to \(\pair_j^{\alpha \beta}(\vec{b}_\ast)\). Since an injective function from \(\pair_j^{\alpha \beta}(\vec{b}_\ast)\) to \(\pair_i^{\alpha \beta}(\vec{b}_\ast)\) can be constructed in the same way, we obtain
\(
    B_i^{\alpha \beta}(\vec{b}_\ast) = B_j^{\alpha \beta}(\vec{b}_\ast)
\),
which proves the independence of \(B_k^{\alpha \beta}(\vec{b}_\ast)\) from \(k\). 

Next, we use this result to prove \eqref{eqn:Fixed_Biweight_Dis}. According to the definition of \(\vec{b}_\ast\), each pair 
\((\vec{c}_1, \vec{c}_2) \in \pair(\vec{b}_\ast)\) has \(b_{\alpha \beta}\) positions, where \(c_{1, i} = \alpha\) and \(c_{2, i} = \beta\) with \(i \in \{1, \dotsc, n\}\), and therefore, is contained in
\(b_{\alpha \beta}\) sets of \(\{\pair_i^{\alpha \beta}(\vec{b}_\ast)\}_i\). Hence, the sum over the cardinalities \(B_i^{\alpha \beta}(\vec{b}_\ast)\) of \(\{\pair_i^{\alpha \beta}(\vec{b}_\ast)\}_i\) is
\begin{align}
    \sum_{i = 0}^{n} B_i^{\alpha \beta}(\vec{b}_\ast) 
    &= b_{\alpha \beta} B(\vec{b}_\ast).
    \label{eqn:Sum1_Weight_Dis}
\intertext{
Since \(B_i^{\alpha \beta}(\vec{b}_\ast)\) is independent of \(i\), we also have 
}
    \sum_{i = 0}^{n} B_i^{\alpha \beta}(\vec{b}_\ast)
    &= n B_k^{\alpha \beta}(\vec{b}_\ast)
    \label{eqn:Sum2_Weight_Dis}
\end{align}
for any \(k \in \{1, \dotsc n\}\).
Solving \eqref{eqn:Sum1_Weight_Dis} and \eqref{eqn:Sum2_Weight_Dis} for \(B_k^{\alpha \beta}(\vec{b}_\ast)\) proves \eqref{eqn:Fixed_Biweight_Dis} of the theorem.

Equation~\eqref{eqn:Fixed_Weight_Dis} follows directly from \eqref{eqn:Fixed_Biweight_Dis} and the identities
\(A(b_1) = B(b_1, 0, 0, n - b_1)\) and \(A_k^{\alpha}(b_1) = B_k^{\alpha\alpha}(b_1, 0, 0, n - b_1)\).
\end{proof}

\section{Decoder Transition Probabilities EaED}
\label{sec:Decoder_Trans_Probs_EaED}
This section presents the derivation of \(\T{\alpha}{\beta}{\Dp, \Ep}\) for the EaED.
In the following, we use \(A \sqcup B\) to denote the union of two disjoint sets \(A\), \(B\) and \(\xor{\beta}\) to negate a binary value \(\beta\).
For a set of codewords \(A \subset \CW\), we define
\(\Sp_t(A) \coloneqq \bigcup_{\vec{c} \in A} \Sp_t(\vec{c})\).

\subsection{Random Experiment}
In decoding steps~\ref{item:GenerateY1} and \ref{item:GenerateY2} of the EaED described in Sec.~\ref{sec:Decoder}, the EaED generates from an error pattern \(\vec{e} \in \Omega\) an error pattern pair \((\vec{e}_1, \vec{e}_2)\).
To describe these pairs, let
\(
    \OmegaP(\alpha_1, \alpha_2) \subset (\ZO^{n})^2
\)
(\(\alpha_1, \alpha_2 \in \{0, 1\}\)) be the set of ordered binary error pattern pairs \((\vec{e}_1, \vec{e}_2)\) for which the following conditions hold:
\begin{itemize}
    \item The distance between \(\vec{e}_1\), \(\vec{e}_2\) at \(\compl{k}\) is 
    \(\operatorname{d}_{\compl{k}}(\vec{e}_1, \vec{e}_2) = \Ep\).
    \item There are \(\Dp\) positions of \(\compl{k}\) at which \(\vec{e}_1\), \(\vec{e}_2\) have a \(1\).
    \item
    \(e_{1, k} = \alpha_1\) and
    \(e_{2, k} = \alpha_2\).
\end{itemize}
Then, the EaED generates from \(\vec{e} \in \Omega\) a pair from the set
\[
    \phantom{.}
    \OmegaP
    \coloneqq 
    \begin{cases}
        \OmegaP(\alpha, \alpha) & \alpha \neq \que \\
        \OmegaP(1, 0) \sqcup \OmegaP(0, 1) & \alpha = \que.
    \end{cases}
\]
It is easy to see that each pair of \(\OmegaP\) occurs with the same probability. Hence, \(\T{\alpha}{\beta}{\Dp, \Ep}\) can be calculated through
\(
    \T{\alpha}{\beta}{\Dp, \Ep} = |\Mset|/|\OmegaP|
\),
where \(\Mset \subset \OmegaP\) contains only the pairs \((\vec{e}_1, \vec{e}_2) \in \OmegaP\) whose decoding result \(\vec{w}\) fulfills \(w_k = \beta\).

\(
    |\OmegaP| = |\Omega| 2^E
\)
because the \(E \coloneqq \Ep + \ind{\alpha = \que}\) erased positions of \(\vec{e}\) can be filled with \(2^E\) different binary values to generate a pair. 
\(|\Mset|\) will be calculated in the following sections.

\subsection{Decomposition of \(|\Mset|\)}
The analysis method in Sec.~\ref{sec:Calc_M3} below can only be applied to specific subsets. Therefore, we first decompose \(\Mset\) into such subsets.

For \(\alpha = \que\), the \(k\)-th bits of the pairs of \(\Mset\) are not fixed. To avoid this, we define 
\(
    \Mset(\alpha_1, \alpha_2) \coloneqq 
    \Mset \cap \OmegaP(\alpha_1, \alpha_2)
\) 
for \(\alpha_1, \alpha_2 \in \ZO\), whose pairs have fixed \(k\)-th bits and get
\begin{equation}\label{eqn:Decomp_M_into_M_alpha1_alpha2}
    \Mset =
    \begin{cases}
        \Mset(\alpha, \alpha) & \alpha \neq \que \\
        \Mset(1, 0) \sqcup \Mset(0, 1) & \alpha = \que.
    \end{cases}
\end{equation}
Below, we will calculate \(|\Mset(\alpha_1, \alpha_2)|\) for arbitrary \(\alpha_1, \alpha_2 \in \ZO\). Then, we can obtain \(|\Mset|\) by \eqref{eqn:Decomp_M_into_M_alpha1_alpha2}.

The following sets \(\Mset_i\) are all subsets of \(\OmegaP(\alpha_1, \alpha_2)\), so for the sake of clarity, we do not specify the domain in the set definitions.
In the case of \(w_k = \beta\), the decoding had to be successful because a failed decoding would result in \(w_k = \alpha \neq \beta\).
That is why \(w_k = \beta\) is only possible for pairs of
\[
    \Mset_1 \coloneqq
    \{
        \text{\(\vec{e}_1 \in \Sp_t(\CW_k^\beta)\) or \(\vec{e}_2 \in \Sp_t(\CW_k^\beta)\)}
    \}.
\]
However, there are pairs \((\vec{e}_1, \vec{e}_2) \in \Mset_1\) where one \(\vec{e}_i\) is closer to a codeword \(\vec{c} \in \CW_k^{\xor{\beta}}\) than the other is to a codeword \(\vec{c} \in \CW_k^{\beta}\). The decoding result \(\vec{w}\) of these pairs fulfills \(w_k = \xor{\beta}\).
Removing these pairs from \(\Mset_1\) results in 
\[
    \Mset(\alpha_1, \alpha_2) 
    = \Mset_1 \setminus
    \underbrace{\{
        (\vec{e}_1, \vec{e}_2) \in \Mset_1, 
        w_k = \xor{\beta}
    \}}_{\eqqcolon \Mset_2}.
\]
Since \(\Mset_2 \subset \Mset_1\), we have \(|\Mset(\alpha_1, \alpha_2)| = |\Mset_1| - |\Mset_2|\).
A further decomposition
\begin{align*}
    \Mset_1 &= 
    \{\vec{e}_1 \in \Sp_t(\CW_k^\beta)\} \cup \{\vec{e}_2 \in \Sp_t(\CW_k^\beta)\}, \\
    \Mset_2 &= 
    \{
        \vec{e}_1 \in \Sp_t(\CW_k^\beta),
        w_k = \xor{\beta}
    \}
    \sqcup
    \{
        \vec{e}_2 \in \Sp_t(\CW_k^\beta), 
        w_k = \xor{\beta}
    \}
\end{align*}
yields
\begin{align*}
    |\Mset_1| &= 2 \cdot 
    |\underbrace{\{
        \vec{e}_1 \in \Sp_t(\CW_k^\beta)   
    \}}_{\eqqcolon \Mset_3}| 
    - 
    |\underbrace{\{
        \vec{e}_1, \vec{e}_2 \in \Sp_t(\CW_k^\beta)
    \}}_{\eqqcolon \Mset_4}|, \\
    |\Mset_2| &= 2 \cdot 
    |\underbrace{\{
        \vec{e}_1 \in S_t(\CW_k^\beta), w_k = \xor{\beta}
    \}}_{\eqqcolon \Mset_5}|.
\end{align*}
because the same decoding method is used for \(\vec{e}_1\) and \(\vec{e}_2\).

\subsection{Derivation of \(|\Mset_3|\)}\label{sec:Calc_M3}
To calculate \(|\Mset_3|\), the following algorithm could be used: Iterate over all triples 
\((\vec{c}, \vec{e}_1, \vec{e}_2)\) with \(\vec{c} \in \CW_k^\beta\)
and count the cases in which \((\vec{e}_1, \vec{e}_2) \in \OmegaP(\alpha_1, \alpha_2)\)
and \(\vec{e}_1 \in \Sp_t(\vec{c})\) holds.%
\footnote{Since the spheres \(\Sp_t(\vec{c_i})\) are pairwise disjoint, no pair is counted twice.}

Since this algorithm is too complex, we use an approach similar to that used in Sec.~\ref{sec:Trans_Prob_EaEDPlus} for the EaED+:
In that section, not all tuples \((\vec{c}, \vec{e})\) are counted individually but all tuples with identical coefficients \(b_{1\que}\), \(b_{10}\), \(b_{01} \coloneqq \Dp - b_{10}\) and \(b_{0\que} \coloneqq \Ep - b_{1\que}\)  could be treated at once.
Figure~\ref{fig:M3_configuration} shows the generalization of this concept for the triple \((\vec{c}, \vec{e}_1, \vec{e}_2)\). 
\begin{figure}
    \centering
    \includegraphics{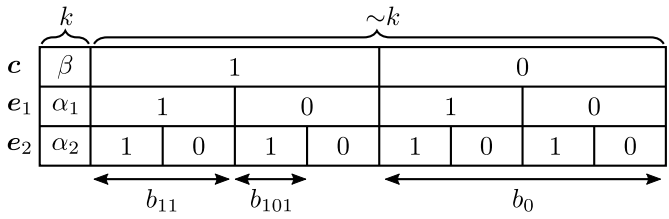}
    \caption{Schematic illustration of the coefficients describing the configuration \(\vec{b}_\ast\) of the triple \((\vec{c}, \vec{e}_1, \vec{e}_2)\).
    For example, \(b_{101}\) determines the number of positions of \(\compl{k}\) where \(\vec{c}\) has a \(1\), \(\vec{e}_1\) has a \(0\) and \(\vec{e}_2\) has a \(1\).}
    \label{fig:M3_configuration}
\end{figure}
Because of \(\vec{c} \in \CW_k^\beta\) and \((\vec{e}_1, \vec{e}_2) \in \OmegaP(\alpha_1, \alpha_2)\), the bits at \(k\) are fixed by \(\beta\), \(\alpha_1\) and \(\alpha_2\).
The overlaps of the bits at \(\compl{k}\) are described by the coefficients
\begin{align*}
    b_{f} &\coloneqq |\{i \in \compl{k} : c_i = f\}|,\\
    b_{fg} &\coloneqq |\{i \in \compl{k} : c_i = f, e_{1, i} = g\}|, \\
    b_{fgh} &\coloneqq |\{i \in \compl{k} : c_i = f, e_{1, i} = g, e_{2, i} = h\}|,
\end{align*}
which are collectively called configuration \(\vec{b}_\ast\).

Let \(\config(\vec{b}_\ast)\) be the set of all triples whose bits at \(k\) are \(\beta\), \(\alpha_1\), \(\alpha_2\) and the positions at \(\compl{k}\) have the configuration \(\vec{b}_\ast\).
Then it is possible to determine if the triples of \(\config(\vec{b}_\ast)\) are counted in the algorithm above, although the exact positions of the symbols are not known:
\begin{subequations}
\label{eqn:Condition_M3_Config}
For \(\vec{e}_1 \in S_t(\vec{c})\), the condition%
\begin{align}
    \dH(\vec{c}, \vec{e}_1) &= 
    b_{10} + b_{01} + \ind{\alpha_1 \neq \beta} \leq t
    \intertext{%
    must hold and for \((\vec{e}_1, \vec{e}_2) \in \OmegaP(\alpha_1, \alpha_2)\), we get
    }
    \operatorname{d}_{\compl{k}}(\vec{e}_1, \vec{e}_2) 
    &= 
    \sum_{\mathclap{v \in \{0, 1\}}} (b_{v10} + b_{v01}) = \Ep,\\
    \sum_{v \in \{0, 1\}} b_{v11} &= \Dp.
\end{align}
\end{subequations}
To calculate \(|\Mset_3|\), we iterate over all configurations fulfilling \eqref{eqn:Condition_M3_Config} and sum the numbers of triples that belong to them:
\begin{align*}
    &|\Mset_3| = \\
    &\sum\limits_{\substack{
        \vec{b}_\ast
        \text{ if \eqref{eqn:Condition_M3_Config}}
    }}
    \underbrace{
        \vphantom{\prod_{v \in \{0, 1\}} \binom{b_v}{b_{v1}}}
        A_k^{\beta}(b_1 + \ind{\beta = 1})
    }_{\text{Ways for \(\vec{c}\)}}
    \underbrace{
        \prod_{v \in \{0, 1\}} \binom{b_v}{b_{v1}}
    }_{\substack{%
        \text{Ways for \(\vec{e}_1\) given \(\vec{c}\)}\\
    }}
    \underbrace{
        \prod_{vv \in \{0, 1\}^2} \binom{b_{vv}}{b_{vv1}}
    }_{\substack{%
        \text{Ways for \(\vec{e}_2\) given \(\vec{c}\), \(\vec{e}_1\)}\\
    }}
    .
\end{align*}

\subsection{Derivation of \(|\Mset_4|\) and \(|\Mset_5|\)}
In the following, \(\beta_1\) and \(\beta_2\) denote binary values, where, in the derivation of \(|\Mset_4|\), we set \(\beta_1 = \beta_2 = \beta\) and in the derivation of \(|\Mset_5|\), \(\beta_1 = \beta\) and \(\beta_2 = \xor{\beta}\).

To calculate \(|\Mset_4|\) and \(|\Mset_5|\), the same approach as before is used for tuples \((\vec{c}_1, \vec{c}_2, \vec{e}_2, \vec{e}_1)\), where position \(k\) is fixed by
\(c_{1, k} = \beta_1\), \(c_{2, k} = \beta_2\), \(e_{2, k} = \alpha_2\) and \(e_{1, k} = \alpha_1\).
The overlaps at \(\compl{k}\) are described by configurations \(\vec{b}_\ast\), whose coefficients have up to \(4\) indices.
As before, the \(i\)-th index in a coefficient denotes the symbol of the \(i\)-th word of \((\vec{c}_1, \vec{c}_2, \vec{e}_2, \vec{e}_1)\).

For \(|\Mset_4|\), we count all tuples for which \((\vec{e}_1, \vec{e}_2) \in \OmegaP(\alpha_1, \alpha_2)\) and \(\vec{e}_1 \in \Sp_t(\vec{c}_1)\), \(\vec{e}_2 \in \Sp_t(\vec{c}_2)\)  holds.
In the configuration domain, \((\vec{e}_1, \vec{e}_2) \in \OmegaP(\alpha_1, \alpha_2)\) transforms into
\begin{subequations}
\label{eqn:Condition_M4_M5}
\begin{align}
    \operatorname{d}_{\compl{k}}(\vec{e}_1, \vec{e}_2) 
    &= 
    \sum_{\mathclap{vv \in \{0, 1\}^2}}
    (b_{vv10} + b_{vv01}) 
    = \Ep,\\
    \sum_{\mathclap{vv \in \{0, 1\}^2}} b_{vv11} 
    &= \Dp
\end{align}
and \(\vec{e}_1 \in \Sp_t(\vec{c}_1)\), \(\vec{e}_2 \in \Sp_t(\vec{c}_2)\) transform into
\begin{align}
    \dH(\vec{c}_1, \vec{e}_1) &= 
    \sum_{\mathclap{vv \in \{0, 1\}^2}}
    (b_{1vv0} + b_{0vv1}) 
    + \ind{\alpha_1 \neq \beta_1} 
    \leq t, \\
    \dH(\vec{c}_2, \vec{e}_2) &=
    \sum_{\mathclap{v \in \{0, 1\}}}
    (b_{v10} + b_{v01})
    + \ind{\alpha_2 \neq \beta_2} 
    \leq t,
\end{align}
\end{subequations}
where \(\beta_1 = \beta_2 = \beta\).
By summing the number of tuples of each configuration that fulfills \eqref{eqn:Condition_M4_M5}, we get
\begin{multline*}
    \phantom{.}
    |\Mset_4| =
    \smash{
        \smashoperator[l]{\sum_{\text{
            \(\vec{b}_\ast\) if \eqref{eqn:Condition_M4_M5}
        }}}
        \Bigg(
    }
    B_k^{\beta \beta}(b_{11} + \ind{\beta = 1}, b_{10}, b_{01}, b_{00} +\ind{\beta = 0}) \\
    \prod_{vv \in \{0, 1\}^2} \binom{b_{vv}}{b_{vv1}}
    \prod_{vvv \in \{0, 1\}^3} \binom{b_{vvv}}{b_{vvv1}}
    \Bigg).
\end{multline*}

For \(|\Mset_5|\), again, only tuples fulfilling \((\vec{e}_1, \vec{e}_2) \in \OmegaP(\alpha_1, \alpha_2)\) and \(\vec{e}_1 \in \Sp_t(\vec{c}_1)\), \(\vec{e}_2 \in \Sp_t(\vec{c}_2)\) are counted, which transforms into \eqref{eqn:Condition_M4_M5} (with \(\beta_1 = \beta\) and \(\beta_2 = \xor{\beta}\)) in the configuration domain.
In addition, 
\begin{equation}
\label{eqn:Condition_M5}
    \dnE{\vec{e}}(\vec{c}_2, \vec{e}) \leq \dnE{\vec{e}}(\vec{c}_1, \vec{e})
\end{equation}
must hold so that the decoder decodes to \(\vec{c}_2\) resulting in \(w_k = \xor{\beta}\).
In configuration domain, these distances are obtained by
\begin{align*}
    \phantom{.}
    \dnE{\vec{e}}(\vec{c}_1, \vec{e}) &= 
    \sum_{\mathclap{v \in \{0, 1\}}} (b_{1v00} + b_{0v11})
    + \ind{\text{\(\alpha \neq \que\) and \(\alpha \neq \beta\)}}, \\
    \dnE{\vec{e}}(\vec{c}_2, \vec{e}) &= 
    \sum_{\mathclap{v \in \{0, 1\}}} (b_{v100} + b_{v011})
    + \ind{\text{\(\alpha \neq \que\) and \(\alpha \neq \xor{\beta}\)}}
    .
\end{align*}
Summing over the valid configurations yields
\begin{multline*}
    |\Mset_5| =
    \smashoperator[l]{\sum_{\text{
        \(\vec{b}_\ast\) if \eqref{eqn:Condition_M4_M5}, \eqref{eqn:Condition_M5}
    }}}\!
    \Bigg(
    B_k^{\beta \xor{\beta}}(b_{11}, b_{10} + \ind{\beta = 1}, b_{01} + \ind{\beta = 0}, b_{00}) \\
    \prod_{vv \in \{0, 1\}^2} \binom{b_{vv}}{b_{vv1}}
    \prod_{vvv \in \{0, 1\}^3} \binom{b_{vvv}}{b_{vvv1}}
    \operatorname{Corr}(\vec{b}_\ast)
    \Bigg).
\end{multline*}
In the case of \(\dnE{\vec{e}}(\vec{c}_2, \vec{e}) = \dnE{\vec{e}}(\vec{c}_1, \vec{e})\), the decoder chooses between \(\vec{c}_1\) and \(\vec{c}_2\) at random, and therefore, on average only half of the pairs have a result with \(w_k = \xor{\beta}\).
To take this into account, the correction term \(\operatorname{Corr}(\vec{b}_\ast)\) is \(\frac{1}{2}\) in this case and \(1\) otherwise.

Now, the transition probabilities can be calculated by putting all contributions together.

\end{document}